\documentclass[10pt,letterpaper]{article}
\usepackage[utf8]{inputenc}
\usepackage{microtype}
\usepackage{amsbsy}
\usepackage{amssymb}
\usepackage{mathrsfs}
\usepackage{MnSymbol}
\usepackage{hyperref}
\usepackage{eufrak}
\usepackage{color}
\usepackage{graphicx}
\usepackage[]{algorithm2e}

\DeclareMathAlphabet{\mathpzc}{OT1}{pzc}{m}{it}

\newtheorem{theorem}{Theorem}
\newtheorem{lemma}{Lemma}
\newtheorem{proposition}{Proposition}

\newtheorem{corollary}{Corollary}
\newenvironment{proof}{{\bf Proof:}}{\hfill $\bot$\medskip }
\newtheorem{example}{Example}
\newtheorem{definition}{Definition}

\def\vars{\mathop{\mathit vars}}
\def\free{\mathop{\mathit free}}
\def\univ{\mathop{\mathit univ}}

\def\sub{\mathop{\mathit sub}}
\def\rep{\mathop{\mathit Rep}}
\def\cert{\mathop{\mathit Cert}}

\def\ar{\mathop{\sl ar}}

\def\lsem{\lbrack\!\lbrack}
\def\rsem{\rbrack\!\rbrack}

\def\false{\mathsf{false}}
\def\true{\mathsf{true}}
\def\I{\mathcal{I}}
\def\J{\mathcal{J}}

\def\rotatecharone#1{\rotatebox[origin=c]{180}{#1}}
\def\creversed{\rotatecharone{{\sf{c}}}}

\def\rotatechartwo#1{\rotatebox[origin=c]{270}{#1}}
\def\cdown{\rotatechartwo{{\sf{c}}}}
\title{Universal (and Existential) Nulls}

\begin{document}
	

\maketitle
\begin{flushleft}
	G\"{o}sta Grahne\textsuperscript{1},
	Ali Moallemi\textsuperscript{2,*}
	\\
	\bigskip
	\bf{1} Concordia University\\
	Montreal, Canada\\
	\texttt{grahne@cs.concordia.ca}
	\\
	\bf{2} Concordia University\\
	Montreal, Canada\\
	* \texttt{moa\_ali@encs.concordia.ca}	
\end{flushleft}
Incomplete Information research is quite mature when
it comes to so called 
{\em existential nulls},
where an existential null is a value stored in the database,
representing an unknown object.
For some reason
{\em universal nulls},
that is, values representing 
{\em all} 
possible objects,
have received almost no attention. We remedy the situation
in this paper,
by showing that a suitable finite representation mechanism,
called 
{\em Star Cylinders},
handling universal nulls can be developed
based on the 
{\em Cylindric Set Algebra} of
Henkin, Monk and Tarski.
We provide a finitary version of the cylindric set algebra,
called 
{\em Cylindric Star Algebra},
and show that our star-cylinders are closed 
under this algebra.
Moreover, we show that any 
{\em First Order Relational Calculus} 
query over databases containing universal nulls can
be translated into an equivalent expression in
our cylindric star-algebra,
and vice versa.
All cylindric star-algebra expressions
can be evaluated 
in time polynomial
in the size of the
database.

The representation mechanism
is then extended to 
{\em Naive Star Cylinders},
which are star-cylinders allowing existential nulls
in addition to universal nulls.
For positive queries (with universal quantification),
the well known naive evaluation technique can still be 
applied on the existential nulls,
thereby allowing polynomial time evaluation
of certain answers on databases
containing both universal and existential nulls.
If precise answers are required,
certain answer evaluation with universal and existential
nulls remains in coNP. 
Note that the problem is coNP-hard, already for positive
existential queries and databases with only existential nulls.
If inequalities $\neg(x_i\approx x_j)$ are allowed,
reasoning over existential databases
is known to be $\Pi^p_2$-complete,
and it remains in $\Pi^p_2$ when
universal nulls and full first order queries
are allowed.


\section{Introduction}\label{intro}

\noindent
In this paper we revisit the foundations of the 
relational model
and unearth {\em universal nulls}, 
showing that they can be treated on par with
the usual {\em existential nulls}
\cite{DBLP:journals/jacm/ImielinskiL84,
	DBLP:conf/pods/FaginKPT09,
        libkin-naive}.
Recall that an existential null in a tuple
in a relation $R$ represents an existentially
quantified variable in an atomic
sentence $R(..)$. This corresponds to the intuition
"value exists, but is unknown."
A universal null,
on the other hand, 
does not represent anything unknown,
but stands for {\em all} values of the domain.
In other words, a universal null represents
a universally quantified variable.
Universal nulls have an obvious application
in databases, as the following example shows.
The symbol "$\ast$" denotes a universal null.

\begin{example}\label{ex1}
Consider binary relations $F${\sl(ollows)}
and $H${\sl (obbies)}, where
$F(x,y)$ means that user $x$ follows
user $y$ on a social media site,
and $H(x,z)$ means that $z$ is a hobby
of user $x$. Let the database be the following.
\begin{center}
\raisebox{-2mm}{
\begin{tabular}{ll}
$F$     &         \\ \hline
Alice   & Chris   \\
$\ast$  & Alice   \\
Bob     & $\ast$  \\
Chris   & Bob     \\
David   & Bob
\end{tabular}
}
$\;\;\;$
\begin{tabular}{ll}
$H$   &            \\ \hline
Alice & Movies     \\
Alice & Music      \\
Bob   & Basketball \\
      &           
\end{tabular}
\end{center}
This is to be interpreted as
expressing the facts that
Alice follows Chris and
Chris and David follow Bob.
Alice is a journalist
who would like to give access to
everyone to articles she shares on
the social media site.
Therefore,
everyone can follow Alice.
Bob is the site 
administrator, 
and is granted the access
to all files anyone shares on the site.
Consequently, Bob follows everyone.
"Everyone" in this context means
all current and possible future   
users.
The query below, 
in domain relational calculus,
asks for the interests of 	
people who are followed by everyone:
\begin{equation}\label{query1}
x_4 \;.\; 
\exists x_2 \exists x_3 \forall x_1
\Big(
F(x_1,x_2)\wedge
H(x_3,x_4)\wedge
(x_2\approx x_3)
\Big).
\end{equation}
The answer to our example query
is 
$\{(\mbox{Movies}),(\mbox{Music})\}$.
Note that star-nulls also can be part
of an answer. For instance,
the query
$x_1,x_2 \;.\; F(x_1,x_2)$
would return all the tuples in $F$.
${\blacktriangleleft}$
\end{example}

Another area of applications of ``*''-nulls
relates to intuitionistic, or constructive
database logic.
In the constructive four-valued approach of 
\cite{amw15} and the
three-valued approach of 
\cite{libkin-naive,libkin-neg} 
the proposition $A \vee \neg A$
is not a tautology.
In order for $A \vee \neg A$ to be true,
we need either a constructive proof of $A$ 
or a constructive proof of $\neg A$. 
Therefore both \cite{amw15} and \cite{libkin-neg}
assume that the database $I$ has a theory of 
the negative information,
i.e.\ that $I=(I^+,I^-)$, where $I^+$ contains the
positive information and $I^-$ the negative information.
The papers \cite{amw15} and \cite{libkin-neg} 
then show how to transform an FO-query $\varphi(\bar{x})$
to a pair of queries $(\varphi^+(\bar{x}),\varphi^-(\bar{x}))$
such that $\varphi^+(\bar{x})$ returns the tuples $\bar{a}$
for which $\varphi(\bar{a})$ is true in $I^+$, and
$\varphi^-(\bar{x})$ returns the tuples $\bar{a}$
for which $\varphi(\bar{a})$ is true in~$I^-$
(i.e.\ $\varphi(\bar{a})$ is false in~$I$).
It turns out that databases containing ``*''-nulls
are suitable for storing~$I^-$.

\begin{example}\label{ex-neg-inf}
Suppose that the instance in Example~\ref{ex1}
represents $I^+$,
and that all negative information we have
deduced	about the $H(obbies)$ relation, 
is that we know Alice doesn't play Volleyball, 
that Bob only has Basketball as hobby, 
and that Chris has no hobby at all.
This negative information about
the relation $H$ is represented by
the table $H^-$ below.
Note that $H^-$ is part of $I^-$.  
\begin{center}
\begin{tabular}{ll}
$H^-$   &            \\ \hline
Alice & Volleyball     \\
Bob & $\ast$ (except Basketball)      \\
Chris   & $\ast$
\end{tabular}
	\end{center}
Suppose the query $\varphi$ asks 
for people who have a hobby, that is
$\varphi = x_1 \,.\, \exists x_2\, H(x_1,x_2)$.
Then $\varphi^+ = \varphi$, and
$\varphi^- = x_1 \,.\, \forall x_2\, H(x_1,x_2)$.	
Evaluating $\varphi^+$ on $I^+$ returns
$\{(\textrm{Alice}),$ $(\textrm{Bob})\}$, and
evaluating $\varphi^-$ on $I^-$ returns
$\{(\textrm{Chris})\}$.
Note that there is no closed-world 
assumption as the negative facts
are explicit.
Thus it is unknown whether David
has a hobby or not.
\end{example}

Universal nulls were first studied in the early
days of database theory by
Biskup in~\cite{DBLP:journals/fuin/Biskup84}.
This was a follow-up on his earlier paper on
existential nulls \cite{DBLP:journals/tods/Biskup83}.
The problem with Biskup's approach,
as noted by himself, was that the semantics
for his algebra worked only for individual operators,
not for compound expressions (i.e.\ queries).
This was remedied in the foundational
paper \cite{DBLP:journals/jacm/ImielinskiL84} 
by Imielinski and Lipski, 
as far as existential
nulls were concerned. Universal nulls
next came up in \cite{ilcyl},
where Imielinski and Lipski showed
that Codd's Relational Algebra
could be embedded in CA,
the {\em Cylindric Set Algebra} 
of Henkin, Monk, and Tarski
\cite{hmt1,hmt2}. As a side remark,
Imielinski and Lipski suggested that the semantics
of their "$\ast$" symbol could be seen as
modeling the universal
null of Biskup.
In this paper we follow their 
suggestion\footnote{We note that  
Sundarmurthy et.\ al.\
\cite{DBLP:conf/icdt/SundarmurthyKLN17}
very recently
have proposed a construct related
to our universal nulls, and studied
ways on placing constraints on them.},
and fully develop a finitary representation
mechanism for databases 
with universal nulls, as well as
an accompanying finitary algebra.
We show that any FO 
(First Order / Domain Relational Calculus)
query 
can be translated into an equivalent
expression in a finitary version of CA,
and that such algebraic expressions can
be evaluated "naively"
by the rules
``$\ast = \ast$''
 and ``$\ast=a$'' for any
constant ``$a$.'' 
Our finitary version is called 
{\em Cylindric Star Algebra} (SCA) and operates
on finite relations containing constants and
universal nulls ``$\ast$.''
These relations are called
{\em Star Cylinders} and they are
finite representations of
a subclass of the infinite cylinders
of Henkin, Monk, and Tarski.
Interestingly, the class of
star-cylinders is closed
under first order querying,
meaning that the infinite result of an FO query
on an infinite instance represented
by a finite sequence of finite star-cylinders can be represented
by a finite star-cylinder.\footnote{Consequently there is
no need to require calculus queries to be ``domain independent.''}
This is achieved by showing that the class of
star-cylinders are closed under our 
cylindric star-algebra, 
and that 
SCA 
as a query language is equivalent
in expressive power with FO.

The Cylindric Set Algebra \cite{hmt1,hmt2} 
---as an algebraization of first order logic--- 
is an algebra on sets of valuations of variables
in an FO-formula. A {\em valuation} $\nu$ of variables
$\{x_1,x_2,\ldots\}$ can be represented as
a tuple $\nu$,
where $\nu(i) = \nu(x_i)$.
The set of all valuations can then be
represented by a 
relation $C$ of such tuples.
In particular, if the FO-formula only
involves a finite number $n$ of variables,
then the representing relation $C$ has arity $n$.
Note however that $C$ has an infinite number
of tuples, since the domain of the variables
(such as the users of a social media site)
should be assumed unbounded.
One of the basic connections \cite{hmt1,hmt2}
between FO and Cylindric Set Algebra is that, 
given any interpretation
$I$ and FO-formula~$\varphi$, 
the set of
valuations $\nu$ under which $\varphi$
is true in $I$ can be represented as such
a relation~$C$. Moreover, each logical
connective and quantifier
corresponds to an operator in the
Cylindric Set Algebra.
Naturally disjunction corresponds to union,
conjunction to intersection,
and negation to complement.
More interestingly, existential quantification
on variable $x_i$ corresponds to 
{\em cylindrification} $\mathsf{c}_i$ on column $i$,
where
$$
\mathsf{c}_i(C)
\;=\;
\{\nu \,:\, \nu(i/a)\in C, \mbox{ for some } a\in\mathbb{D}\}, 
$$
and $\nu(i/a)$ denotes the valuation (tuple)
$\nu'$, where $\nu'(i)=a$ and $\nu'(j)=\nu(j)$
for $i\neq j$.
The algebraic counterpart of
universal quantification can be derived
from cylindrification and complement,
or be defined directly as 
{\em inner cylindrification}
$$
\creversed_i(C)
\;=\;
\{\nu \,:\, \nu(i/a)\in C, \mbox{ for all } a\in\mathbb{D}\}. 
$$

In addition, in order to represent equality,
the Cylindric Set Algebra
also contains constant relations $d_{ij}$
representing the equality $x_i\approx x_j$.
That is, $d_{ij}$ is the set of all valuations
$\nu$, such that $\nu(i)=\nu(j)$.

The objects $C$ and $d_{ij}$ of \cite{hmt1,hmt2}
are of course infinitary.
In this paper we therefore develop a
finitary representation mechanism,
namely relations containing
universal nulls ``$\ast$'' and
certain equality literals.
These objects are called
{\em Star Tables} when they 
represent the records stored in
the database.
When used as run-time constructs
in algebraic query evaluation,
they will be called {\em Star Cylinders}.
Example \ref{ex1} showed star-tables in
a database. The run-time variable binding
pattern of the query \eqref{query1},
as well as its algebraic evaluation
is shown in the star-cylinders
in Example \ref{ex2} below.    

\begin{example}\label{ex2}
Continuing Example \ref{ex1}, 
in that database
the atoms $F(x_1,x_2)$ 
and $H(x_3,x_4)$ of query \eqref{query1}
are represented by
star-tables $C_F$ and $C_H$,
and the equality atom is 
represented by the star-cylinder $C_{23}$.
Note that these are positional relations,
the "attributes" $x_1,x_2,x_3,x_4$
are added for illustrative purposes only. 

\begin{center}
\begin{tabular}{llll}
$C_F$     &         \\ \hline
$x_1$   & $x_2$  & $x_3$ & $x_4$     \\ \hline
Alice   & Chris   & $\ast$ & $\ast$  \\
$\ast$  & Alice   & $\ast$ & $\ast$  \\
Bob     & $\ast$  & $\ast$ & $\ast$  \\
Chris   & Bob     & $\ast$ & $\ast$
\end{tabular}

\medskip

\hspace*{3ex}
\begin{tabular}{llll}
$C_H$   &        & & \\ \hline
$x_1$   & $x_2$  & $x_3$ & $x_4$     \\ \hline
$\ast$ & $\ast$ & Alice & Movies     \\
$\ast$ & $\ast$ & Alice & Music      \\
$\ast$ & $\ast$ & Bob   & Basketball \\
\end{tabular}

\medskip

\begin{tabular}{lllll}
$C_{23}$     &   &      \\ \hline
$x_1$   & $x_2$  & $x_3$ & $x_4$ & \\ \hline
$\ast$  & $\ast$ & $\ast$ & $\ast$ & 2=3
\end{tabular}
\end{center}
The algebraic translation of query \eqref{query1}
is the SCA-expression
\begin{eqnarray}\label{q1alg}
\dot{\mathsf{c}}_2(
\dot{\mathsf{c}}_3(\dot{\creversed}_{1}((C_F \capdot C_H) \capdot C_{23})))
\end{eqnarray}
The intersection of $C_F$ and $C_H$
is carried out as star-intersection $\capdot$,
where for instance 
$\{(a,*,*)\} \capdot \{(*,b,*)\}
=
\{(a,b,*)\}$.
The result will contain 12 tuples,
and when these are star-intersected with $C_{23}$,
the star-cylinder $C_{23}$ will act as a selection
by columns 2 and 3 being equal. The result
is the star-cylinder $C' = (C_F \capdot C_H) \capdot C_{23}$
below. 
\begin{center}
\begin{tabular}{llll}
$C'$ &&&\\ \hline
$x_1$   & $x_2$   & $x_3$  &  $x_4$     \\ \hline
$\ast$  & Alice   & Alice  & Movies     \\
$\ast$  & Alice   & Alice  & Music      \\
Bob     & Alice   & Alice  & Movies     \\
Bob     & Alice   & Alice  & Music      \\
Bob     & Bob     & Bob    & Basketball \\
Chris   & Bob     & Bob    & Basketball
\end{tabular}
\end{center}
The inner star-cylindrification
on column 1 then yields
$C'' = \dot{\creversed}_{1}((C_F \capdot C_H) \capdot C_{23}).$
\begin{center}
\begin{tabular}{llll}
$C''$ &&& \\ \hline
$x_1$   & $x_2$   & $x_3$  &  $x_4$     \\ \hline
$\ast$  & Alice   & Alice  & Movies     \\
$\ast$  & Alice   & Alice  & Music
\end{tabular}
\end{center}
Finally,
applying outer star-cylindrifications on columns 2 and 3
of star-cylinder $C''$
yields the final result
$C''' = \dot{\mathsf{c}}_2(
\dot{\mathsf{c}}_3(\dot{\creversed}_{1}((C_F \capdot C_H) \capdot C_{23}))).$
\begin{center}
\begin{tabular}{llll}
$C'''$ &&& \\ \hline
$x_1$   & $x_2$   & $x_3$  &  $x_4$     \\ \hline
$\ast$  & $\ast$  & $\ast$ & Movies     \\
$\ast$  & $\ast$  & $\ast$ & Music
\end{tabular}
\end{center}
The system can now return the answer,
i.e.\ the values of column 4 in
cylinder $C'''$.  
Note that columns where all rows are ``$\ast$''
do not actually have to be materialized
at any stage. Negation requires some
additional details that will be introduced
in Section \ref{adding-neg}.
$\blacktriangleleft$
\end{example}

The aim of this paper is to 
develop a clean and sound modelling of
universal nulls, and furthermore show that
the model can be seamlessly extended 
to incorporate the existential nulls
of Imielinski and Lipski
\cite{DBLP:journals/jacm/ImielinskiL84}.
We show that 
FO and our SCA are equivalent in expressive
power when it comes to querying databases
containing universal nulls,
and that SCA queries can be evaluated
(semi) naively.
This will be done in three steps:
In Section \ref{foandca}
we show the equivalence
between FO and Cylindric Set Algebra over
infinitary databases. This was of course 
only the starting point of \cite{hmt1,hmt2},
and we recast the result here in terms
of database theory.\footnote{Van Den Bussche
\cite{jvdb} has recently referred to
\cite{hmt1,hmt2} in similar terms.}
In Section \ref{caandcastar}
we introduce our finitary Cylindric Star Algebra.
Section \ref{only-pos} develops
the machinery for the positive case,
where there is no negation in the query or database.
This is then extended to include negation 
in Section \ref{adding-neg}.
By these two sections we 
show that certain
infinitary cylinders can be finitely represented
as star-cylinders, and that our finitary 
Cylindric Star Algebra 
on finite star-cylinders
mirrors the Cylindric Set Algebra on the
infinite cylinders they represent.
In Section \ref{stored} 
we tie these two results
together, delivering the promised
SCA evaluation of FO queries on databases
containing universal nulls.
In Section \ref{naive}
we seamlessly extend our framework
to also handle existential nulls,
and show that naive evaluation
can still be used for positive queries
(allowing universal quantification,
but not negation) 
on databases containing both universal
and existential nulls.
Section \ref{complexity} 
then shows that all 
SCA expressions
can be evaluated in time polynomial
in the size of the database
when only universal nulls are present.
We also show that when both universal and existential
nulls are present, the certain answer to any
negation-free (allowing inner cylindrification,
i.e.\ universal quantification) SCA-query 
can be evaluated naively in polynomial time.
When negation is present 
it has long been known that
the problem is coNP-complete
for databases containing existential nulls.
We show that the problem remains coNP-complete
when universal nulls are allowed in addition
to the existential ones.
For databases containing
existential nulls it has been known that 
database containment and view containment
are
coNP-complete and $\Pi^p_2$-complete,
respectively.
We also show that the addition of universal nulls
does not increase these complexities.

%
%
\section{Relational calculus and\\ 
cylindric set algebra}\label{foandca}
Throughout this paper we assume a
fixed schema $\mathscr{R} = \{R_1,\ldots,R_m,\approx\}$,
where each $R_p$,
$p\in\{1,\ldots,m\}$,
is a relational symbol
with an associated
positive integer $\ar(R_p)$,
called the {\em arity} of $R_p$.
The symbol $\approx$ represents equality.

\bigskip
\noindent
{\bf Logic.}
Our calculus is the  standard domain relational calculus.
Let $\{x_1,x_2,\ldots\}$ be a countably infinite set
of {\em variables}.
We define the set of {\em FO-formulas} 
$\varphi$ (over $\mathscr{R}$)
in the usual way:
$R_p(x_{i_{1}},\ldots,x_{i_{\ar(R_{p})}})$
and
$x_i\approx x_j$
are atomic formulas, and these are closed under
$\wedge,\vee,\neg,\exists x_i,$ and
$\forall x_i,$ 
in a well-formed manner
possibly using parenthesis's for disambiguation.

Let $\varphi$ be an FO-formula.
We denote by $\vars(\varphi)$
the set of variables in $\varphi$,
by $\free(\varphi)$
the set of free variables in $\varphi$,
and by $\sub(\varphi)$ the set of 
subformulas of $\varphi$
(for formal definitions, see \cite{DBLP:books/aw/AbiteboulHV95}).
If $\varphi$ has $n$ variables
we say that $\varphi$ is an {\em FO$_n$-formula}. 
We assume without loss of generality
that each variable occurs only once in the formula,
except in equality literals, and that a formula
with $n$ variables uses variables
$x_1,\ldots,x_n$. 

\bigskip
\noindent
{\bf Instances.}
Let $\mathbb{D} = \{a_1,a_2,\ldots\}$ 
be a countably infinite {\em domain}.
An {\em instance} $I$ (over $\mathscr{R}$)
is a mapping that assigns a possibly infinite subset
$R_p^I$ of $\mathbb{D}^{\ar(R_{p})}$ 
to each relation symbol $R_p$,
and $\approx^I \;=\; \{(a,a) \,:\, a\in \mathbb{D}\}$.
Note that our instances are infinite model-theoretic ones.
The set of tuples actually recorded in the database will be called
the {\em stored database} (to be defined in Section \ref{stored}).

In order to define the (standard) notion
of truth of an FO$_n$-formula $\varphi$ in an instance $I$
we first define a {\em valuation} to be a mapping
$\nu : \{x_1,\ldots,x_n\} \rightarrow \mathbb{D}$.
If $\nu$ is a valuation, $x_i$ a variable and $a\in\mathbb{D}$,
then $\nu_{(i/a)}$ denotes the valuation
which is the same as $\nu$, 
except $\nu_{(i/a)}(x_i)=a$.
Then we use the usual recursive definition
of $I\models_{\nu}\varphi$,
meaning 
{\em instance $I$ satisfying $\varphi$
under valuation $\nu$},
i.e.\
$I\models_{\nu}(x_i\approx~x_j)$ if 
$(\nu(x_i),\nu(x_j))\in\; \approx^I$, 
$I\models_{\nu}R_p(x_{i_{1}},\ldots,x_{i_{\ar(R_p)}})$ if
$(\nu(x_{i_{1}}),\ldots,\nu(x_{i_{\ar(R_{p})}}))\in~R_{p}^{I}$, and    
$I\models_{\nu}\exists x_i\, \varphi$ if
$I\models_{\nu_{(i/a)}}\varphi$ for some $a\in\mathbb{D}$,
and so on.
Our stored databases will be finite representations
of infinite instances, so the semantics of answers
to FO-queries 
will be defined in terms of the infinite instances:
\begin{definition}\label{fo1}
Let $I$ be an instance,
and $\varphi$ an FO$_n$-formula with 
$\free(\varphi) = \{x_{i_{1}},\ldots,x_{i_{k}}\}$, $k\leq n$.
Then the \textit{answer to $\varphi$ on $I$} is defined as
$$
\varphi^I 
=
\{
(\nu(x_{i_{1}}),\ldots,\nu(x_{i_{k}}))
\;:\;
I\models_{\nu}\varphi
\}.
$$
\end{definition}

\bigskip
\noindent
{\bf Algebra.}
As noted in \cite{ilcyl} the relational algebra
is really a disguised version of the 
Cylindric Set Algebra of Henkin,
Monk, and Tarski \cite{hmt1,hmt2}.
We shall therefore work directly with
the Cylindric Set Algebra instead of
Codd's Relational Algebra.
Apart from the conceptual clarity,
the Cylindric Set Algebra will also allow
us to smoothly introduce the promised universal nulls.

Let $n$ be a fixed positive integer.
The basic building block of the Cylindric Set Algebra
is an {\em $n$-dimensional cylinder}
$C\subseteq\mathbb{D}^n$.
Note that a cylinder is essentially an infinite
$n$-ary relation. They will however be called
cylinders,
in order to distinguish them from
instances. The rows in a cylinder will represent
run-time variable valuations, 
whereas tuples in instances represent
facts about the real world.
We also have special cylinders 
called {\em diagonals},
of the form 
$d_{ij} = \{t\in\mathbb{D}^n \;:\; t(i)=t(j)\}$
representing  the equality $x_i \approx x_j$.
We can now define the Cylindric Set Algebra.

\begin{definition}\label{ca1}
Let $C$ and $C'$ be infinite $n$-dimensional cylinders.
The {\em Cylindric Set Algebra} consists of the
following operators.

\begin{enumerate}
\item
{\em Union:}
$C \,\bigcup\, C'$.
Set theoretic union.

\item
{\em Complement:}
$\overline{C} = \mathbb{D}^n\setminus C$.
			
\item
{\em Outer cylindrification:}
$$
\mathsf{c}_{i}(C) = \{t\in\mathbb{D}^n \;:\; t(i/a)\in C, 
\mbox{ for some } a\in\mathbb{D}\}.
$$
\end{enumerate}
\end{definition}

\noindent
The operation
$\mathsf{c}_i$ is 
called outer cylindrification on the $i$:th dimension,
and will correspond to existential quantification of
variable $x_i$.
For the geometric intuition behind the name cylindrification,
see \cite{hmt1,ilcyl}.
Intersection is considered a derived operator,
and we also introduce
the following derived operator:

\begin{enumerate}\setcounter{enumi}{3}
\item
{\em Inner cylindrification:}
$\creversed_{i}(C) = \overline{\mathsf{c}_i(\overline{C})}$,
corresponding to universal quantification.
Note that 
$$
\creversed_{i}(C) 
= \{t\in \mathbb{D}^n \;:\; t(i/a)\in C, \mbox{ for all } a\in\mathbb{D}\}.
$$
\end{enumerate}

\medskip
We also need the notion of
cylindric set algebra expressions.

\begin{definition}\label{caexpr}
Let $\mathbf{C} = (C_1,\ldots,C_m, d_{ij})_{i,j\,\in\,\{1,\ldots,n\}}$
be a sequence of infinite $n$-dimensional cylinders and diagonals.
The set of \textit{CA$_n$-expressions} 
(over $\mathbf{C}$)
is obtained
by closing the atomic expressions 
$\mathsf{C}_p$ and $\mathsf{d}_{ij}$
under union, intersection, complement, and 
inner and outer cylindrifications.
Then $E(\mathbf{C})$,
the \textit{value of expression $E$ 
on sequence $\mathbf{C}$}
is defined in the usual way,
e.g.\ $\mathsf{C}_p(\mathbf{C}) = C_p$,
$\mathsf{d}_{ij}(\mathbf{C}) = d_{ij}$,
$\mathsf{c}_{i}(E)(\mathbf{C}) =
\mathsf{c}_{i}(E(\mathbf{C}))$ etc.
\end{definition}

\bigskip

\noindent
{\bf Equivalence of FO and CA.}
In the next two theorems
we will restate, 
in the context of the relational model,
the correspondence between domain relational calculus and
cylindric set algebra as query languages on instances
\cite{hmt1,hmt2}.
An expression $E$ in cylindric set algebra
of dimension $n$ will be called a CA$_n$-expression.
When translating
an FO$_n$-formula to a CA$_n$-expression
we first need to extend all $k$-ary relations in $I$
to $n$-ary
by filling the $n-k$ last 
columns in all possible ways.
Formally, this is expressed as follows:
\begin{definition}\label{hor1}
The \textit{horizontal $n$-expansion} 
of an infinite $k$-ary relation $R$
is 
$$
\mathsf{h}^n(R)
=
\bigcup_{t\in R}\;\{t\}\times\mathbb{D}^{n-k}.
$$
The equality relation 
$\approx^I = \{(a,a) : a\in \mathbb{D}\}$
is expanded into diagonals
$d_{ij}$ for $i,j\in\{1\ldots,n\}$,
where 
$$
d_{ij} = \bigcup_{(a,a)\in\approx^{I}}
\mathbb{D}^{i-1}\times\{a\}\times\mathbb{D}^{j-i+1}\times\{a\}\times\mathbb{D}^{n-j},
$$ 
and for an instance
$I = (R_1^I,\ldots,R_m^I,\approx^I)$, we have
$$
\mathsf{h}^n(I) = 
(\mathsf{h}^n(R_1^I),\ldots,\mathsf{h}^n(R_m^I),d_{ij})_{i,j}.
$$
\end{definition}
Once an instance is expanded it becomes
a sequence 
$\mathbf{C} = (C_1,\ldots,C_m, d_{ij})_{i,j}$
of $n$-dimensional cylinders and
diagonals, on which Cylindric Set Algebra 
Expressions can be applied.

The main technical difficulty in
the translation from FO$_n$ to CA$_n$ 
is the correlation of the variables in the
FO$_n$-sentence $\varphi$
with the columns in the expanded
relations in the instance.
This can be achieved using a derived 
``swapping'' operator
$\mathsf{z}^{i_1,\ldots,i_k}_{j_1,\ldots,j_k}$
that interchanges the columns
$i_l$ and $j_l$,
where
$l\in\{1,\ldots,k\}$.\footnote{This 
was already implicitly done in the expansion
of $\approx^I$ in Definition \ref{hor1}.
For a definition of swapping using
the primitive operators, see 
Definition 1.5.12 in \cite{hmt1}.}
Every atom $R_p$ in $\varphi$
will correspond to a CA$_n$-expression 
$C_p = \mathsf{h}^n(R_p^I)$.
However,
for every occurrence of an atom
$R_p(x_{i_{1}},\ldots,x_{i_{k}})$ in $\varphi$
we need to interchange the columns 
$1,\ldots,k$ 
with columns $i_1,\ldots,i_k$. 
This is achieved by the expression
$\mathsf{z}^{1,\ldots,k}_{i_{1},\ldots,i_{k}}(\mathsf{C}_p)$.

\medskip
Among the many identities holding in Cylindric Set Algebra
we will in the sequel need the following ones
\begin{proposition}\label{propz}
	\cite{hmt1}.
	Let $C$ be an $n$-dimensional cylinder,
	and $i,j\in\{1,\ldots,n\}$.
	Then 
\end{proposition}
\vspace*{-3mm}
\begin{enumerate}
	\item
	$
	\mathsf{z}^{i}_{j}(C)=
	\mathsf{z}_{i}^{j}(C).
	$
	
	\item
	$
	\mathsf{z}^{i}_{j}(
	\mathsf{z}^{j}_{i}(C))=C.
	$

	\item
	$
	\mathsf{c}_i(\mathsf{z}^i_j(C))
	\;=\;
	\mathsf{z}^i_j(\mathsf{c}_j(C)).
	$

	\item
	If $i\neq j$ then
	$
	\mathsf{z}^{i}_{j}
	(C \setminus C')
	\;=\;
	\mathsf{z}^{i}_{j}(C) 
	\setminus\,
	\mathsf{z}^{i}_{j}(C').
	$
	\item
	If $\mathsf{c}_i(C)=C$
	and $\mathsf{c}_j(C)=C$
	then
	$
	\mathsf{z}^i_j(C)
	\;=\;
	C.
	$
\end{enumerate}

\begin{proposition}\label{propzz}
	Let $i,j,k$ be pairwise
	distinct natural numbers,
	such that
	$\{i,j,k\} \cap \{1,2,3\} = \emptyset$,
	and let $C$ be an $n$-dimensional cylinder
	that is 2-full\footnote{Cylinder $C$ is $i$-full if 
	$\mathsf{c}_i(C)=C$.} and $k$-full. Then
	$$
	\mathsf{z}^{i,k}_{1,2}
	(
	\mathsf{z}^{3,2,1}_{k,j,i}(C)
	)
	\;=\;
	\mathsf{z}^{1,2,3}_{1,j,2}(C).
	$$
\end{proposition}
\begin{proof}
	$$
	\mathsf{z}^{i,k}_{1,2}
	(
	\mathsf{z}^{3,2,1}_{k,j,i}(C)
	)
	= 
	\mathsf{z}^{i,k,3,2,1}_{1,2,k,j,i}(C)
	= 
	\mathsf{z}^{i,3,2,2,1}_{1,2,k,j,i}(C)
	= \mathsf{z}^{i,3,2,1}_{1,2,j,i}(C)
	= 
	$$
	\vspace{-.4cm}
	$$
	\hspace{1.8cm}
	\mathsf{z}^{i,1,3,2}_{1,i,2,j}(C)
	= 
	\mathsf{z}^{i,1,3,2}_{i,1,2,j}(C)
	= 
	\mathsf{z}^{1,3,2}_{1,2,j}(C).
	$$
	
	\noindent
	The second equality follows from
	Theorem 1.5.18 in \cite{hmt1},
	the third equality holds since
	$\mathsf{c}_2(C) = C$ and
	$\mathsf{c}_k(C) = C$,
	the fourth since
	$\{1,i\}\cap \{2,3,j\}= \emptyset$.
	The last two equalities follow from
	Theorem 1.5.17 and
	1.5.13 in \cite{hmt1},
	respectively.   
\end{proof}

The entire FO$_n$-formula $\varphi$ 
with $\free(\varphi) = \{x_{i_{1}},\ldots,x_{i_{k}}\}$
will then
correspond  to the CA$_n$-expression
$E_{\varphi} = \mathsf{z}^{i_{1},\ldots,i_{k}}_{1,\ldots,k}(F_{\varphi})$,
where $F_{\varphi}$ is defined
recursively as follows:
\begin{itemize}
\item
If $\varphi = R_p(x_{i_{1}},\ldots,x_{i_{k}})$
where $k=\ar(R_p)$,
then
$
F_{\varphi}
= 
\mathsf{z}^{1,\ldots,k}_{i_{1},\ldots,i_{k}}(
\mathsf{C}_p).
$

\item
If $\varphi = x_i\approx x_j$,
then
$
F_{\varphi}
= 
\mathsf{d}_{ij}
$.

\item
If $\varphi = \psi\vee\chi$,
then
$F_{\varphi} = F_{\psi} \bigcup F_{\chi}$,
if $\varphi = \psi\wedge\chi$,
then
$F_{\varphi} = F_{\psi} \bigcap F_{\chi}$,
and if $\varphi = \neg\, \psi$, 
then $F_{\varphi} =\overline{ F_{\psi}}$.

\item
If $\varphi = \exists x_{i} \psi$,
then $F_{\varphi} = \mathsf{c}_{i}(F_{\psi})$.

\item
If $\varphi = \forall x_{i} \psi$,
then $F_{\varphi} = \creversed_{i}(F_{\psi})$.
\end{itemize}
For an example,
let us reformulate the $FO_4$-query $\varphi$ from 
\eqref{query1} as 
\begin{equation*}
x_4 \;.\; 
\exists x_2 \exists x_3 \forall x_1\; 
\Big( R_1(x_1,x_2)\wedge R_2(x_3,x_4)\wedge(x_2\approx x_3)\Big)
\end{equation*}
When translating $\varphi$
the relation $R_1^I$ is first expanded to 
$C_1 = R_1^I\times\mathbb{D}\times\mathbb{D}$,
and
$R_2^I$ is expanded to 
$C_2 = R_2^I\times\mathbb{D}\times\mathbb{D}$.
In order to correlate the variables in $\varphi$
with the columns in the expanded databases,
we do the shifts
$\mathsf{z}^{1,2}_{1,2}(C_1)$
and
$\mathsf{z}^{1,2}_{3,4}(C_2)$.
The equality $(x_2\approx x_3)$
was expanded to the diagonal
$d_{23} = \{t\in\mathbb{D}^n \,:\, t(2)=t(3)\}$
so here the variables are already correlated.
After this the conjunctions are replaced
with intersections and the
quantifiers with cylindrifications.
Finally, the column corresponding to the
free variable $x_4$ 
in $\varphi$ (whose bindings
will constitute the answer) is shifted to
column 1.
The final CA$_n$-expression will then
be evaluated against~$I$ as
$E_{\varphi}(\mathsf{h}^4(I)) =$
\begin{equation*}
\mathsf{z}^{4}_{1}
\Big(
\mathsf{c}_{23}(
\creversed_{1}(
\mathsf{z}^{1,2}_{1,2}(
R_1^I\times\mathbb{D}^2)
\,\bigcap\,
\mathsf{z}^{1,2}_{3,4}(
R_2^I\times\mathbb{D}^2)
\,\bigcap\,
d_{23}
))
\Big).
\end{equation*}
We now have 
$E_{\varphi}(\mathsf{h}^4(I)) = \mathsf{h}^4(\varphi^I)$.
The following fundamental
result follows from
\cite{hmt1,hmt2},
but we prove it here
for the benefit of the readers
who don't want to consult \cite{hmt1,hmt2}.
\begin{theorem}\label{thrm1}
For all FO$_n$-formulas $\varphi$,
there is a CA$_n$ expression~$E_{\varphi}$,
such that
$$
E_{\varphi}(\mathsf{h}^n(I))
\;=\;
\mathsf{h}^n(\varphi^{I}),
$$
for all instances $I$.
$\blacktriangleleft$
\end{theorem}
\begin{proof}
	We prove the stronger claim:
	For all FO$_n$-formulas $\varphi$,
	for all $\psi\in\sub(\varphi)$,
	with $\free(\psi)=\{x_{i_{1}},\ldots,x_{i_{k}}\}$,
	there is an CA$_n$ expression~$E_{\psi}$,
	such that
	$$
	\mathsf{z}_{1,\ldots,k}^{i_{1},\ldots,i_{k}}
	\Big(
	E_{\psi}(\mathsf{h}^n(I))
	\Big)
	\;=\;
	\mathsf{h}^n(\psi^{I}),
	$$
	for all instances $I$.
	The main claim the follows since
	$\varphi\in\sub(\varphi)$,
	and the outermost sequence of
	swappings can be considered part
	of the final expression $E_{\varphi}$.
	In all cases below we assume wlog\footnote{
		If $k=n$ we can introduce an additional variable
		$x_{n+1}$ and the conjunct
		$\exists x_{n+1}(x_{n+1}\approx x_{n+1})$
		which would assure that the $n+1$:st dimension is full.
		Alternatively, we could introduce swapping as a primitive
		in the algebra. This however would require a corresponding
		renaming operator in the FO-formulas, see~\cite{hmt1}.}
	that $k<n$ so that the $k+1$:st column can be used
	in the necessary swappings.
	
	\begin{itemize}
		\item
		$\psi = R_p(x_{i_{1}},\ldots,x_{i_{k}})$,
		where $k=\ar(R_p)$. 
		We let
		$
		E_{\psi}
		= 
		\mathsf{z}^{k,\ldots,1}_{i_{k},\ldots,i_{1}}(
		\mathsf{C}_p).
		$
		We have
		\begin{center}
			$
			\begin{array}{lrr}
			\mathsf{z}_{1,\ldots,k}^{i_{1},\ldots,i_{k}}
			\Big(
			E_{\psi}(\mathsf{h}^n(I))
			\Big)
			& = & \\ \\
			\mathsf{z}_{1,\ldots,k}^{i_{1},\ldots,i_{k}}
			\Big(
			\mathsf{z}^{k,\ldots,1}_{i_{k},\ldots,i_{1}}
			(
			\mathsf{C}_p
			(
			\mathsf{h}^n(I)
			)
			\Big)
			& = &\\
			& \textit{\small{By Proposition \ref{propz}} (2)}&\\
			\mathsf{C}_p
			(
			\mathsf{h}^n(I)
			)
			& = & \\ \\
			\mathsf{h}^n(R_{p}^{I})
			& = \\ \\
			\mathsf{h}^n(\psi^I).
			&  & \\ \\ 
			\end{array}
			$
		\end{center}

		\item
		$\psi =x_i\approx x_j$.
		We assume wlog that $n>2$ so that swaps
		can be performed.
		We let
		$
		E_{\psi}
		= 
		\mathsf{d}_{ij}
		$.
		We then have
		\begin{center}
			$
			\begin{array}{lrr}
			\mathsf{z}^{i,j}_{1,2}
			\Big(
			E_{\psi}(\mathsf{h}^n(I))
			\Big)
			& = & \\ \\
			\mathsf{z}^{i,j}_{1,2}
			\Big(
			\mathsf{d}_{ij}
			\Big)
			& = & \\ \\
			\mathsf{z}^{i,j}_{1,2}
			\Big(
			\{t\in\mathbb{D}^n \;:\; t(i)=t(j)\}
			\Big)
			& = & \\ \\
			\{t\in\mathbb{D}^n \;:\; t(1)=t(2)\}
			& = & \\ \\
			\{(a,a) \;:\; a\in\mathbb{D}\}\times\mathbb{D}^{n-2}
			& = & \\ \\
			\mathsf{h}^n(
			\{(a,a) \;:\; a\in\mathbb{D}\}
			)
			& = & \\ \\
			\mathsf{h}^n((x_i\approx x_j)^I)
			& =  \\ \\
			\mathsf{h}^n(\psi^I).
			& & \hspace*{3cm}   \\ \\ 
			\end{array}
			$
		\end{center}
		\bigskip
		
		\item
		$\psi = \neg\, \xi$, with
		$\free(\xi) = \{x_{i_{1}},\ldots,x_{i_{k}}\}$.
		We assume wlog that $k<n$.
		Then $E_{\psi} =\overline{ E_{\xi}}$,
		and the inductive hypothesis is
		$$
		\mathsf{z}_{1,\ldots,k}^{i_{1},\ldots,i_{k}}
		\Big(
		E_{\xi}(\mathsf{h}^n(I))
		\Big)
		\;=\;
		\mathsf{h}^n(\xi^{I}) 
		$$
		
		We have
		
		\begin{center}
			$
			\begin{array}{lrr}
			\mathsf{z}_{1,\ldots,k}^{i_{1},\ldots,i_{k}}
			\Big(E_{\psi}(\mathsf{h}^n(I))
			\Big)
			& = &\\ \\
			\mathsf{z}_{1,\ldots,k}^{i_{1},\ldots,i_{k}}
			\Big(\overline{
				E_{\xi}(\mathsf{h}^n(I))
			}\Big)
			& = &\\ \\
			\mathsf{z}_{1,\ldots,k}^{i_{1},\ldots,i_{k}}
			\Big(\mathbb{D}^n\setminus
			E_{\xi}(\mathsf{h}^n(I))
			\Big)
			& =  &\\
			&  \textit{\small{By Proposition \ref{propz}} (2)}& \\ 
			\mathsf{z}_{1,\ldots,k}^{i_{1},\ldots,i_{k}}
			\Big(\mathbb{D}^n\setminus
			(
			\mathsf{z}_{i_{k},\ldots,i_{1}}^{k,\ldots,1}
			(
			\mathsf{z}_{1,\ldots,k}^{i_{1},\ldots,i_{k}}
			(
			E_{\xi}(\mathsf{h}^n(I))
			)
			)
			)\Big)
			& = & \\ \\
			\mathsf{z}_{1,\ldots,k}^{i_{1},\ldots,i_{k}}
			\Big(\mathbb{D}^n\setminus
			(
			\mathsf{z}_{i_{k},\ldots,i_{1}}^{k,\ldots,1}
			(
			\mathsf{h}^n(\xi^{I})
			)
			)\Big)
			& =  &\\
			&  \textit{\small{By Proposition \ref{propz}} (5)}& \\
			\mathsf{z}_{1,\ldots,k}^{i_{1},\ldots,i_{k}}
			\Big(
			\mathsf{z}_{i_{k},\ldots,i_{1}}^{k,\ldots,1}
			(\mathbb{D}^n) \setminus
			(
			\mathsf{z}_{i_{k},\ldots,i_{1}}^{k,\ldots,1}
			(
			\mathsf{z}^n(\xi^{I})
			)
			)\Big)
			& =  &\\
			&  \textit{\small{By Proposition \ref{propz}} (4)}& \\
			\mathsf{z}_{1,\ldots,k}^{i_{1},\ldots,i_{k}}
			\Big(
			\mathsf{z}_{i_{k},\ldots,i_{1}}^{k,\ldots,1}
			(\mathbb{D}^n \setminus
			\mathsf{h}^n(\xi^{I})
			)
			)\Big)
			& =  &\\
			&  \textit{\small{By Proposition \ref{propz}} (2)}& \\
			D^n \setminus
			\mathsf{h}^n(\xi^{I})
			& = & \\ \\
			\mathsf{h}^n((\neg\, \xi)^{I})
			& =  & \\ \\
			\mathsf{h}^n(\psi^{I}).
			&&  \\ \\
			\end{array}
			$
		\end{center}


		\item
		$\psi = \xi\wedge\chi$,
		with $\free(\psi)=\{x_{i_{1}},\ldots,x_{i_{k}}\}$,
		$\free(\xi)=\{x_{r_{1}},\ldots,x_{r_{p}}\}$,
		$\free(\chi)=\{x_{s_{1}},\ldots,x_{s_{q}}\}$,
		$\free(\psi) = \free(\xi)\cup\free(\chi)$,
		and\footnote{The last assumption is needed in steps $\dagger$}
		$\free(\xi)\cap\free(\chi)=\emptyset$.
		Now $E_{\psi} = E_{\xi} \bigcap E_{\chi}$.
		The inductive hypothesis is
		$$
		\mathsf{z}_{1,\ldots,p}^{r_{1},\ldots,r_{p}}
		\Big(
		E_{\xi}(\mathsf{h}^n(I))
		\Big)
		\;=\;
		\mathsf{h}^n(\xi^{I}).
		$$
		$$
		\mathsf{z}_{1,\ldots,q}^{s_{1},\ldots,s_{q}}
		\Big(
		E_{\chi}(\mathsf{h}^n(I))
		\Big)
		\;=\;
		\mathsf{h}^n(\chi^{I}).
		$$
		
		We have
		
		\begin{center}
			$
			\begin{array}{lrl}
			\mathsf{z}_{1,\ldots,k}^{i_{1},\ldots,i_{k}}
			\Big(
			E_{\psi}(\mathsf{h}^n(I))
			\Big)
			& = & \\ \\
			\mathsf{z}_{1,\ldots,k}^{i_{1},\ldots,i_{k}}
			\Big(
			E_{\xi}\bigcap E_{\chi}\, (\mathsf{h}^n(I))
			\Big)
			& = & \\ \\
			\mathsf{z}_{1,\ldots,k}^{i_{1},\ldots,i_{k}}
			\Big(
			E_{\xi}(\mathsf{h}^n(I))
			\,\bigcap\;
			E_{\chi}(\mathsf{h}^n(I))
			\Big)
			& =  &\\
			&  \textit{\small{By Proposition \ref{propz}} (2)}& \\
			\mathsf{z}_{1,\ldots,k}^{i_{1},\ldots,i_{k}}\Big(
			\mathsf{z}_{r_{p},\ldots,r_{1}}^{p,\ldots,1}
			\Big(
			\mathsf{z}_{1,\ldots,p}^{r_{1},\ldots,r_{p}}
			(
			E_{\xi}(\mathsf{h}^n(I))
			)
			\Big)
			\;\bigcap\;
			\mathsf{z}_{s_{q},\ldots,s_{1}}^{q,\ldots,1}
			\Big(
			\mathsf{z}_{1,\ldots,q}^{s_{1},\ldots,s_{q}}
			(
			E_{\chi}(\mathsf{h}^n(I))
			)
			\Big)\Big)
			& = & \\ \\
			\mathsf{z}_{1,\ldots,k}^{i_{1},\ldots,i_{k}}\Big(
			\mathsf{z}_{r_{p},\ldots,r_{1}}^{p,\ldots,1}
			\Big(\mathsf{h}^n(\xi^I)
			\Big)
			\;\bigcap\;
			\mathsf{z}_{s_{q},\ldots,s_{1}}^{q,\ldots,1}
			\Big(
			\mathsf{h}^n(\chi^I)
			\Big)\Big)
			& = & \\ \\
			\mathsf{z}_{1,\ldots,k}^{i_{1},\ldots,i_{k}}\Big(
			&&\\
			\hspace*{1.5cm}\mathsf{z}_{r_{p},\ldots,r_{1}}^{p,\ldots,1}
			\Big(\mathsf{h}^n(
			\{\nu(x_{r_{1}}),\ldots,\nu(x_{r_{p}})
			\;:\;
			I\models_{\nu}\xi\}
			)
			\Big)
			\;\;\bigcap\;
			&&\\
			\hspace*{1.5cm}\mathsf{z}_{s_{q},\ldots,s_{1}}^{q,\ldots,1}
			\Big(
			\mathsf{h}^n(
			\{\nu(x_{s_{1}}),\ldots,\nu(x_{s_{q}})
			\;:\;
			I\models_{\nu}\chi\}
			)
			\Big)
			&&\\
			\hspace{1.2cm}\Big)
			& =\dagger  &\\
			&  \textit{\small{By Proposition \ref{propz}} (5)}& \\
			\mathsf{z}_{1,\ldots,k}^{i_{1},\ldots,i_{k}}\Big(
			&&\\
			\hspace*{1.5cm}\mathsf{z}_{s_{q},\ldots,s_{1}, r_p, \ldots,r_1}^{\:p+q,\ldots,\:p+1, \: p, \ldots, 1}
			\Big(\mathsf{h}^n(
			\{\nu(x_{r_{p}}),\ldots,\nu(x_{r_{1}})
			\;:\;
			I\models_{\nu}\xi\}
			)
			\Big)
			\;\;\bigcap\;
			&&\\
			\hspace*{1.5cm}
			\mathsf{z}_{r_{p},\ldots,r_{1},s_q, \ldots,s_1}^{q+p,\ldots,q+1, q, \ldots, 1}
			\Big(
			\mathsf{h}^n(
			\{\nu(x_{s_{1}}),\ldots,\nu(x_{s_{q}})
			\;:\;
			I\models_{\nu}\chi\}
			)
			\Big)
			&&\\
			\hspace{1.2cm}\Big)
			& =\dagger & \\ \\
			\mathsf{z}_{1,\ldots,k}^{i_{1},\ldots,i_{k}}\Big(
			\mathsf{z}_{i_{k},\ldots,i_{1}}^{k,\ldots,1}\Big(
			\mathsf{h}^n(
			\{\nu(x_{i_{1}}),\ldots,\nu(x_{i_{k}})
			\;:\;
			I\models_{\nu}\xi\wedge\chi
			\}
			)
			\Big)
			\Big)
			& =  &\\
			&  \textit{\small{By Proposition \ref{propz}} (2)}& \\
			\mathsf{h}^n(
			\{\nu(x_{i_{1}}),\ldots,\nu(x_{i_{k}})
			\;:\;
			I\models_{\nu}\xi\wedge\chi
			\}
			& = & \\ \\
			\mathsf{h}^n(
			(\xi\wedge\chi)^I
			)
			& = &    \\ \\
			\mathsf{h}^n(\psi^I).
			&
			\end{array}
			$
		\end{center}
		
		\newpage
		
		
		\item
		$\psi = \exists x_{i_{j}} \xi$, with
		$\free(\xi) = \{x_{i_{1}},\ldots,x_{i_{j}},\ldots,x_{i_{k}}\}$.
		Let 
		
		\begin{center}
			$
			\begin{array}{lrl}
			\{i'_1,\ldots,i'_{k-1}\} & = &
			\{i_1,\ldots,i_j,\ldots,i_{k}\} 
			\setminus
			\{i_j\}\\
			\{r_1,\ldots,r_{n-k}\} & = &
			\{1,\ldots,n\} \setminus
			\{i_1,\ldots,i_j,\ldots,i_{k}\} \\
			\{r'_1,\ldots,r'_{n-k+1}\} & = &
			\{r_1,\ldots,r_{n-k}\}
			\cup
			\{i_j\} 
			\end{array}
			$
		\end{center}
		
		\noindent
		We assume wlog that $k<n$. 
		Let $E_{\psi} = \mathsf{c}_{i_{j}}(E_{\xi})$.
		The inductive hypothesis is
		$$
		\mathsf{z}_{1,\ldots,k}^{i_{1},\ldots,i_{k}}
		\Big(
		E_{\xi}(\mathsf{h}^n(I))
		\Big)
		\;=\;
		\mathsf{h}^n(\xi^{I}). 
		$$
		
		We have
		
		\begin{center}
			$
			\begin{array}{lrl}
			\mathsf{z}_{1,\ldots,k-1}^{i'_{1},\ldots,i'_{k-1}}
			\Big(E_{\psi}(\mathsf{h}^n(I))
			\Big)
			& = &\\ \\
			\mathsf{z}_{1,\ldots,k-1}^{i'_{1},\ldots,i'_{k-1}}
			\Big(\mathsf{c}_{i_{j}}(
			E_{\xi}(\mathsf{h}^n(I))
			)\Big)
			& = & \\
			& \textit{\small{By Prop.\ \ref{propz}} (3)} & \\
			\mathsf{z}_{1,\ldots,k-1}^{i'_{1},\ldots,i'_{k-1}}
			\Big(\mathsf{c}_{i_{j}}(
			\mathsf{z}^{k,\ldots,1}_{i_{k},\ldots,i_{1}}
			(
			\mathsf{z}_{1,\ldots,k}^{i_{1},\ldots,i_{k}}
			(
			E_{\xi}(\mathsf{h}^n(I))
			)
			)
			)\Big)
			& = &\\ \\
			\mathsf{z}_{1,\ldots,k-1}^{i'_{1},\ldots,i'_{k-1}}
			\Big(\mathsf{c}_{i_{j}}(
			\mathsf{z}^{k,\ldots,1}_{i_{k},\ldots,i_{1}}
			(
			\mathsf{h}^n(\xi^{I})
			)
			)\Big)
			& = & \\
			& \textit{\small{By Prop.\ \ref{propz}} (3)}& \\
			\mathsf{z}_{1,\ldots,k-1}^{i'_{1},\ldots,i'_{k-1}}
			\Big(
			\mathsf{z}^{k,\ldots,1}_{i_{k},\ldots,i_{1}}
			(
			\mathsf{c}_{j}(
			\mathsf{h}^n(\xi^{I})
			)
			)\Big)
			& = &\\ \\
			\mathsf{z}_{1,\ldots,j-1,j,\ldots,k-1}^{i_{1},\ldots,i_{j-1},i_{j+1},\ldots,i_{k}}
			\Big(
			\mathsf{z}^{k,\ldots,j,j-1,\ldots,1}_{i_{k},\ldots,i_j,i_{j-1},\ldots,i_{1}}
			(
			\mathsf{c}_{j}(
			\mathsf{h}^n(\xi^{I})
			)
			)\Big)
			& =& \\ \\
			\mathsf{z}_{1,\ldots,j-1}^{i_{1},\ldots,i_{j-1}}
			\circ 
			\mathsf{z}_{j,\ldots,k-1}^{i_{j+1},\ldots,i_{k}}
			\Big(
			\mathsf{z}^{k,\ldots,j+1,j}_{i_{k},\ldots,i_{j+1},i_{j}}
			\circ 
			(\mathsf{z}^{j-1,\ldots,1}_{i_{j-1},\ldots,i_{1}}
			\mathsf{c}_{j}(
			\mathsf{h}^n(\xi^{I})
			)
			)\Big)
			& =  & \\
			& \textit{\small{{By Prop.\ \ref{propzz}}}} &\\
			\mathsf{z}^{1,\ldots,j-1,i_j,j,\ldots,k-1}
			_{1,\ldots,j-1,j,j+1,\ldots,k}(
			\mathsf{c}_{j}(\mathsf{h}^n(\xi^{I})
			))
			& =  & \\
			& \textit{\small{By Prop.\ \ref{propz}} (3)} &\\
			\mathsf{c}_{i_{j}}(
			\mathsf{z}^{1,\ldots,j-1,i_j,j,\ldots,k-1}
			_{1,\ldots,j-1,j,j+1,\ldots,k}(
			\mathsf{h}^n(\xi^{I})
			))
			& = &\\ \\
			\mathsf{c}_{i_{j}}(
			\mathsf{z}^{1,\ldots,j-1,i_j,j,\ldots,k-1}
			_{1,\ldots,j-1,j,j+1,\ldots,k}(
			\mathsf{h}^n(
			\{(\nu(x_{i_{1}}),\ldots,
			\nu(x_{i_{j}}),\ldots,
			\nu(x_{i_{k}})) \;:\; I\models_{\nu}\xi\}
			)
			))
			& = &\\ \\
			\mathsf{c}_{i_{j}}(
			\mathsf{z}^{1,\ldots,j-1,i_j,j,\ldots,k-1}
			_{1,\ldots,j-1,j,j+1,\ldots,k}(&&\\
			\qquad \qquad
			\{(\nu(x_{i_{1}}),\ldots,
			\nu(x_{i_{j}}),\ldots,
			\nu(x_{i_{k}}),\nu(x_{r_{1}}),
			\ldots,\nu(x_{r_{n-k}})
			) \;:\; I\models_{\nu}\xi\}
			))
			& = &\\ \\
			
			\end{array}
			$
		\end{center}
		\begin{center}
			$
			\begin{array}{lrr}
			\mathsf{c}_{i_{j}}(
			\{(\nu(x_{i'_{1}}),\ldots,
			\nu(x_{i'_{k-1}}),\nu(x_{r'_{1}}),\ldots,
			\nu(x_{i_{j}}),\ldots,
			\ldots,\nu(x_{r'_{n-k+1}})
			) \;:\; I\models_{\nu}\xi\}
			)
			& = \\ \\
			\bigcup_{a\in D}
			\{(\nu(x_{i'_{1}}),\ldots,
			\nu(x_{i'_{k}}),\nu(x_{r'_{1}}),
			\ldots,\nu(x_{i_{j}}),\ldots
			\nu(x_{r'_{n-k+1}}),
			) \;:\; I\models_{\nu_{(i_{j}/a)}}\!\xi\}
			& = \\ \\
			\{(\nu(x_{i'_{1}}),\ldots,
			\nu(x_{i'_{k}}),\nu(x_{r'_{1}}),\ldots,
			\nu(x_{i_{j}}),
			\ldots,\nu(x_{r'_{n-k+1}})
			) \;:\; I\models_{\nu}\exists x_{i_{j}}\xi\}
			& = \\ \\
			\mathsf{h}^n(
			\{(\nu(x_{i'_{1}}),\ldots,
			\nu(x_{i'_{k-1}})
			) \;:\; I\models_{\nu}\exists x_{i_{j}}\xi\}
			)
			& = \\ \\
			\mathsf{h}^n(
			\xi^I
			).
			&
			\end{array}
			$
		\end{center}
	\end{itemize}
\end{proof}


\noindent
On the other hand,
CA$_n$ expressions $E$ are translated into FO$_n$-formulas
$\varphi_E$ recursively as follows:

\begin{itemize}
\item
If $E = \mathsf{C}_p$, 
then 
$$
\varphi_E 
\;=\; 
R_p(x_1,\ldots,x_{\ar(R_{p})}) 
\;\wedge\!\!\!\!
\bigdoublewedge_{k\,\in\,\{\!\ar(R_{p})+1,\ldots,n\}}
(x_k \approx x_k).
$$

\item
If $E = \mathsf{d}_{ij}$,
then 
$$
\varphi_E \;=\;
(x_i\approx x_j)
\;\wedge\!\!\!\!
\bigdoublewedge_{k\;\in\,\{1,\ldots,n\}\setminus\{i,j\}}
(x_k \approx x_k).
$$

\item
If $E = F \bigcup G$, then
$\varphi_E = \varphi_F \vee \varphi_G$,
if $E = F \bigcap G$, then
$\varphi_E = \varphi_F \wedge \varphi_G$,
and if $E = \overline{F}$,
then $\varphi_E = \neg \varphi_F$.

\item
If $E = \mathsf{c}_i(F)$,
then $\varphi_E = (\exists x_i \varphi_F) \wedge (x_i \approx x_i)$.

\item
If $E = \creversed_i(F)$,
then $\varphi_E = (\forall x_i \varphi_F) \wedge (x_i \approx x_i)$.
\end{itemize}

\medskip
\noindent
The following result 
can also be extracted from
\cite{hmt1,hmt2}.

\begin{theorem}
For every CA$_n$ expression $E$
there is an FO$_{n}$ formula $\varphi_E$,
such that
$$
\varphi_E^I
\;=\;
E(\mathsf{h}^{n}(I)),
$$
for all instances~$I$.
$\blacktriangleleft$
\end{theorem}
\begin{proof}
	We do a structural induction
	
	\begin{itemize}
		\item
		$E = \mathsf{C}_p$. 
		Then 
		$
		\varphi_E 
		\;=\; 
		R_p(x_1,\ldots,x_k) 
		\;\wedge\;
		\bigdoublewedge_{r\,\in\,\{k+1,\ldots,n\}}
		(x_r \approx x_r),
		$
		where $k=\ar(R_p)$.
		Clearly
		\begin{center}
			$
			\begin{array}{lr}
			\varphi_E^I 
			& = \\ \\
			\{
			(\nu(x_1),\ldots,\nu(x_k),\nu(x_{k+1}),\ldots,\nu(x_n))
			\;:\;
			I\models_{\nu}R_p(x_1,\ldots,x_k)
			\}
			& = \\ \\
			R_p^I\times\mathbb{D}^{n-k}
			& = \\ \\
			\mathsf{C}_p(\mathsf{h}_n(I))
			& = \\ \\
			E(\mathsf{h}^n(I)).
			\end{array}
			$
		\end{center}
		
		\bigskip
		
		\item
		$E = \mathsf{d}_{ij}$. 
		Then 
		$
		\varphi_E 
		\;=\; 
		(x_i\approx x_j)
		\;\wedge\;
		\bigdoublewedge_{r\;\in\,\{1,\ldots,n\}\setminus\{i,j\}}
		(x_r \approx x_r)
		$.
		We have
		\begin{center}
			$
			\begin{array}{lr}
			\varphi_E^I
			& = \\ \\
			\{
			(\nu(x_1),\ldots,
			\nu(x_i),\ldots,\nu(x_j),\ldots,
			\nu(x_n))
			:
			I\models_{\nu} (x_i\approx x_j)
			\}
			& = \\ \\
			\{t\in\mathbb{D}^n :\; t(i)=t(j)\}
			& = \\ \\
			\mathsf{d}_{ij}
			& = \\ \\
			E(\mathsf{h}^n(I)).
			\end{array}
			$
		\end{center}
		\bigskip
		
		\item
		$E = F_1 \,\bigcap\, F_2$.
		Then $\varphi_E = \varphi_{F_{1}} \wedge \varphi_{F_{2}}$,
		and the inductive hypothesis is
		
		\begin{center}
			$
			\begin{array}{lrr}
			\varphi_{F_{1}}^{I}  &=& 
			F_1(\mathsf{h}^n(I))
			\\
			\varphi_{F_{2}}^{I} &=& 
			F_2(\mathsf{h}^n(I))
			\end{array}
			$
		\end{center}
		\noindent
		Then, 
		\begin{center}
			$
			\begin{array}{lr}
			\varphi_E^I
			& = \\ \\
			(\varphi_{F_{1}} \wedge \psi_{F_{2}})^I
			& = \\ \\
			\{
			(\nu(x_1),\ldots,\nu(x_{n}))
			\;:\;
			I\models_{\nu}\varphi_{F_{1}}\!\wedge\,\psi_{F_{2}}
			\}
			& = \\ \\
			\{
			(\nu(x_1),\ldots,\nu(x_{n}))
			\;:\;
			I\models_{\nu}\varphi_{F_{1}}
			\}
			\;\;\cap &\\
			\{
			(\nu(x_{1}),
			\ldots,\nu(x_{n}))
			\;:\;
			I\models_{\nu}\xi_{F_{2}}
			\}
			& = \\ \\
			\varphi_{F_{1}}^I
			\cap\,
			\xi_{F_{2}}^I
			& = \\ \\
			F_1(\mathsf{h}^n(I))
			\,\bigcap\;
			F_2(\mathsf{h}^n(I))
			& = \\ \\
			F_1 \bigcap\, F_2\;(\mathsf{h}^n(I)) 
			& = \\ \\
			E(\mathsf{h}^n(I)). 
			\end{array}
			$
		\end{center}
		\bigskip
		

		\item
		$E = \overline{F}$,
		where
		Then $\varphi_E = \neg\varphi_F$,
		and
		the inductive hypothesis is
		$
		\varphi_F^I
		=
		F(\mathsf{h}^n(I)).
		$	
		We have
		\begin{center}
			$
			\begin{array}{lr}
			\varphi_E^I
			& = \\ \\
			\neg\varphi_F^I
			& = \\ \\ 
			\overline{\varphi_F^I}
			& = \\ \\ 
			\overline{F(\mathsf{h}^n(I))}
			& = \\ \\
			E(\mathsf{h}^n(I)).
			& 
			\end{array}
			$
		\end{center}
		\bigskip	
		
		\bigskip
		
		\item
		$E = \mathsf{c}_{i}(F)$, 
		Then $\varphi_E = (\exists x_{i}\, \varphi_F)
		\wedge (x_i \approx x_i) $.
		The inductive hypothesis is
		$
		\varphi_F^I
		=
		F(\mathsf{h}^n(I)).
		$	

		We have		
		\begin{center}
			$
			\begin{array}{lr}
			\varphi_E^I
			& = \\ \\
			\{
			(\nu(x_1),\ldots,\nu(x_i),\ldots,\nu(x_n))
			\;:\;
			I\models_{\nu}(\exists x_{i}\, \varphi_F)
			\wedge (x_i \approx x_i) 
			\}
			& = \\ \\
			\{
			(\nu(x_1),\ldots,\nu(x_i),\ldots,\nu(x_n))
			\;:\;
			I\models_{\nu}(\exists x_{i}\, \varphi_F)
			\}
			\;\;\;\cap 
			&  \\
			\{
			(\nu(x_1),\ldots,\nu(x_i),\ldots,\nu(x_n))
			\;:\;
			I\models_{\nu}(x_i \approx x_i) 
			\}
			& = \\ \\
			\{
			(\nu(x_1),\ldots,\nu(x_i),\ldots,\nu(x_n))
			\;:\;
			I\models_{\nu}(\exists x_{i}\, \varphi_F)
			\}
			\;\;\;\cap\;\;\; 
			\mathbb{D}^n
			& = \\ \\
			\{
			(\nu(x_1),\ldots,\nu(x_i),\ldots,\nu(x_n))
			\;:\;
			I\models_{\nu}(\exists x_{i}\, \varphi_F)
			\}
			& = \\ \\
			\bigcup_{a\in\mathbb{D}}\;
			\{(\nu((x_{1}),\ldots,
			\nu(x_{i}),\ldots,
			\nu(x_{n})) 
			\;:\; 
			I\models_{\nu_{(i\!/a)}}\varphi_F
			\}
			& = \\ \\
			\mathsf{c}_{i}(
			\{(\nu((x_{1}),\ldots,
			\nu(x_{i}),\ldots,
			\nu(x_{n})) 
			\;:\; 
			I\models_{\nu}\varphi_F\} 
			)
			& = \\ \\ 
			\mathsf{c}_{i}(
			\varphi_F^I
			)
			& = \\ \\ 
			\mathsf{c}_{i}(
			F(
			\mathsf{h}^n(
			I
			)))
			& = \\ \\
			E\,(\mathsf{h}^n(I)).
			\end{array}
			$
		\end{center}
		
	\end{itemize}
\end{proof}

%

\section{Cylindric Set Algebra and \\
Cylindric Star Algebra}\label{caandcastar}
Since cylinders can be infinite,
we want a finite mechanism to represent
(at least some) infinite cylinders,
and the mechanism to be closed under queries.
Our representation mechanism comes in
two variations, depending on whether negation
is allowed or not. We first consider
the positive (no negation) case.

\subsection{Positive framework}\label{only-pos}

\noindent
{\bf Star Cylinders.}
We define an 
{\em $n$-dimensional (positive) star-cylinder}
$\dot{C}$ to be a finite set of
{\em $n$-ary star-tuples},
the latter being elements of
$(\mathbb{D}\cup\{*\})^n \times\, \powerset(\Theta_n)$,
where $\Theta_n$ denotes the set of all
{\em equalities}
of the form $i=j$,
with $i,j\in\{1,\ldots,n\}$.
Star-tuples will be denoted $\dot{t}, \dot{u}, \ldots$,
where a  star-tuple such as
$\dot{t} = (a,*,c,*,*, \{(4=5)\})$,
is meant to represent the set of
all ``ordinary" tuples
$(a,x,c,y,y)$ where $x,y\in\mathbb{D}$.
It will be convenient to assume that all
our star-cylinders are in the following normal form.

\begin{definition}\label{normalform}
An $n$-dimensional star-cylinder $\dot{C}$ is said to be
in {\em normal form} if $\dot{t}(n+1)$ 
is satisfiable, and 
$\dot{t}(n+1)\models(i=j)$ entails
$(i=j)\in\dot{t}(n+1)$ and $\dot{t}(i) = \dot{t}(j)$,
for all star-tuples $\dot{t}\in\dot{C}$. 
\end{definition}

The symbol $\models$ above 
stands for standard logical implication.
It is easily seen that maintaining star-cylinders
in normal form can be done efficiently in polynomial time.
We shall therefore assume without loss of generality
that all star-cylinders and star-tuples are in normal form.
We next define the notion of {\em dominance},
where a dominating star-tuple represents
a superset of the ordinary tuples represented
by the dominated star-tuple. 
First we define a relation 
$\preceq \;\;\subseteq\; (\mathbb{D}\cup\{*\})^2$
by $a\preceq a$,
$*\preceq *$,
and $a\preceq *$,
for all $a\in\mathbb{D}$.
\begin{definition}\label{domin}
Let $\dot{t}$ and $\dot{u}$ be 
$n$-dimensional star-tuples.
We say that $\dot{u}$ 
{\em dominates} $\dot{t}$, 
denoted
$\dot{t} \preceq \dot{u}$, 
if $\dot{t}(i)\preceq\dot{u}(i)$
for all $i\in\{1,\ldots,n\}$,
and $(i=j)\in\dot{u}(n+1)$
entails $(i=j)\in\dot{t}(n+1)$
when $\dot{t}(i)=\dot{t}(j)=*$, and
entails $\dot{t}(i)=\dot{t}(j)$ otherwise.
\end{definition}

We complete the definition by stipulating
that $\dot{t}\preceq\dot{u}$ whenever
$\dot{t}(n+1)\models\false$\footnote{Note that
in this case (only), $\dot{t}$ is not in normal form.}.
We can now define the {\em meet} 
$\,\dot{t}\curlywedge\dot{u}\,$ of star-tuples
$\dot{t}$ and $\dot{u}\;$:

\begin{definition}\label{meet}
Let $\dot{t}$ and $\dot{u}$ be $n$-ary star-tuples.
If $\dot{t}(j),\dot{u}(j)\in\mathbb{D}$ 
for some $j$
and $\dot{t}(j)\neq\dot{u}(j)$ 
then $\dot{t}\curlywedge\dot{u}\,(i) = a$
for $i\in\{1,\ldots,n\}$,   
and $\dot{t}\curlywedge\dot{u}(n+1) = \false$.   
\footnote{Here $a$ is an arbitrary constant in $\mathbb{D}$.}
Otherwise, for $i\in\{1,\ldots,n\}$
$$
{\dot{t}}\curlywedge \dot{u}\,(i) = 
\left\{
	\begin{array}{ll}
	\dot{t}(i)  & \mbox{ if } \;\dot{t}(i)\in\mathbb{D} \\
	\dot{u}(i)  & \mbox{ if } \;\dot{u}(i)\in\mathbb{D} \\
	*           & \mbox{ if } \;\dot{t}(i)=\dot{u}(i)=*
	\end{array}  
	\right.
$$
and 
\begin{eqnarray*}
\dot{t}\curlywedge{\dot{u}}\,(n+1) &=& \dot{t}(n+1)\cup\dot{u}(n+1).
\end{eqnarray*}
\end{definition}

For an example, let 
$\dot{t} = (a,*,*,*,*,\{(3=4)\})$ and
$\dot{u} = (*,b,*,*,*,\{(4=5)\})$. 
Then we have
$\dot{t}\curlywedge\dot{u} = (a,b,*,*,*,\{(3=4),(4=5),(3=5)\})$.
Note that
$
\dot{t}\curlywedge\dot{u}
\preceq 
\dot{t} 
$, and
$
\dot{t}\curlywedge\dot{u}
\preceq 
\dot{u} 
$.
Note also that
for $n$-ary star-tuples
$\dot{t}_{\emptyset} = (a, a, \ldots, a, \{\false\})$
and
$\dot{t}_{\mathbb{D}^{n}} = (*, *, \ldots, *, \{\true\})$,
and for any $n$-ary star-tuple $\dot{t}$,
it holds that
$\dot{t}\curlywedge\dot{t}_{\emptyset} = \dot{t_{\emptyset}}$,
$\dot{t}\curlywedge\dot{t}_{\mathbb{D}^{n}} = \dot{t}$, and 
$\dot{t}_{\emptyset} \preceq 
\dot{t} \preceq 
\dot{t}_{\mathbb{D}^{n}}$.

We extend the order $\preceq$ to include
"ordinary" $n$-ary tuples 
$t\in\mathbb{D}^n$ by 
identifying $(a_1,\ldots,a_n)$ with
star-tuple $(a_1,\ldots,a_n,\{\true\})$.
Let $\dot{C}$ be an $n$-dimensional star-cylinder.
We can now define the meaning of $\dot{C}$ to be
the set $\lsem\dot{C}\,\rsem$ of all ordinary tuples it represents,
where
$$
\lsem\dot{C}\,\rsem
=\;
\{t\in\mathbb{D}^n :\; t\preceq\dot{u}
\mbox{ for some }
\dot{u}\in\dot{C} 
\}.
$$
We lift the order to $n$-dimensional star-cylinders
$\dot{C}$ and $\dot{D}$, by stipulating that
$\dot{C}\preceq\dot{D}$,
if for all star-tuples 
$\dot{t}\in\dot{C}$
there is a star-tuple
$\dot{u}\in\dot{D}$,
such that 
$\dot{t}\preceq\dot{u}$.

\begin{lemma}
Let $\dot{C}$ and $\dot{D}$
be $n$-dimensional (positive) star-cylinders.
Then 
$\lsem\dot{C}\rsem
\subseteq 
\lsem\dot{D}\rsem$
iff
$\dot{C}\preceq\dot{D}$.
\end{lemma}
\smallskip

\noindent
\begin{proof}
We first show that
$\lsem\{\dot{t}\}\rsem
\subseteq
\lsem\dot{D}\rsem$ iff
there is a star-tuple
$\dot{u}\in\dot{D}$,
such that $\dot{t}\preceq\dot{u}$.
For a proof, we note that 
if $\dot{t}\preceq\dot{u}$
for some $\dot{u}\in\dot{D}$,
then 
$\lsem\{\dot{t}\,\}\rsem
\subseteq
\lsem\dot{D}\rsem$.
For the other direction, 
assume that $\lsem\{\dot{t}\,\}\lsem\,\subseteq\lsem\dot{D}\rsem$.
Let $A\subseteq\mathbb{D}$ be the finite set
of constants appearing in $\dot{t}$ or $\dot{D}$.
Construct the tuple 
$t\in(A\cup\{*\})^n$,
where $t(i)=\dot{t}(i)$ if $\dot{t}(i)\in A$,
and
$t(i)=a_i$ if $\dot{t}(i)=*$.
Here $a_i$ is
a unique value in 
the set $\mathbb{D}\setminus A$.
If $\dot{t}(n+1)$ contains an equality $(i=j)$
we choose $a_i=a_j$. 
Then
$t\in\lsem\{\dot{t}\,\}\rsem\subseteq\lsem\dot{D}\rsem$,
so there must be a tuple $\dot{u}\in\dot{D}$,
such that $t\preceq\dot{u}$.
It remains to show that
$\dot{t}\preceq\dot{u}$.
If $t(i)=a$ for some $a\in A$,
then $\dot{t}(i)=a$,
and since $t\preceq\dot{u}$ it follows that
$\dot{t}(i)\preceq\dot{u}(i)$.
If $t(i) = a_i\notin A$ then $\dot{t}(i)=*$,
and therefore
$t(i/b)\in\lsem\{\dot{t}\}\rsem\subseteq\lsem\dot{D}\rsem$,
for any $b$ in the infinite set $\mathbb{D}\setminus{A}$.
Therefore
it must be that $\dot{u}(i)=*$,
and thus $\dot{t}(i)\preceq\dot{u}(i)$.
This is true for all
$i\in\{1,\ldots,n\}$.
Finally, if $(i=j)\in\dot{u}(n+1)$,
we have two cases:
If
$t(i)\in A$
then $\dot{t}(i)=\dot{t}(j)$,
and if $t(i)\notin A$
then $(i=j)\in\dot{t}(n+1)$.
In summary, we have shown that 
$\dot{t}\preceq\dot{u}$.

We now return to the proof of the claim of the lemma.
The if-direction follows directly from definitions.
For the only-if direction,
assume that
$\lsem\dot{C}\rsem\subseteq\lsem\dot{D}\rsem$.
To see that $\dot{C}\preceq\dot{D}$
let $\dot{t} \in \dot{C}$.
Then 
$\lsem\{\dot{t}\,\}\rsem
\subseteq\lsem\dot{C}\rsem
\subseteq\lsem\dot{D}\rsem$.
We have just shown above that this entails
that there is a $\dot{u}\in\dot{C}$
such that $\dot{t}\preceq\dot{u}$,
meaning that
$\dot{C}\preceq\dot{D}$.
\end{proof}

\subsubsection*{Positive Cylindric Star Algebra}
Next we redefine the positive cylindric set operators
so that 
$\lsem\dot{C} \;\dot{\circ}\; \dot{D}\rsem
=\;
\lsem\dot{C}\rsem \circ\; \lsem\dot{D}\rsem$
or 
$\circ(\lsem\dot{D}\rsem) 
\;=\;
\lsem\dot{\circ}(\dot{D})\rsem$,
for each positive cylindric operator $\circ$,
its redefinition~$\dot{\circ}$,
and star-cylinders $\dot{C}$ and $\dot{D}$.
\begin{definition}
The positive cylindric star-algebra consists of
the following operators.
\end{definition}
\begin{enumerate}
\item
{\em Star-diagonal:}
$\dot{d}_{ij} = \{(\,\overbrace{*,\ldots,*}^{n},(i=j))\}$ 

\item
{\em Star-union:}
$
\dot{C}\cupdot\dot{D} 
= 
\{\dot{t} \::\: \dot{t}\in\dot{C} \mbox{ or } \dot{t}\in\dot{D}\}
$

\item
{\em Star-intersection:}
$
\dot{C}\capdot\dot{D} 
= 
\{\dot{t}\curlywedge\dot{u} 
\::\: 
\dot{t}\in\dot{C} \mbox{ and } \dot{u}\in\dot{D}\}
$

\item
{\em Outer cylindrification:}
Let $i\in\{1,\ldots,n\}$,
let $\dot{C}$ be an $n$-dimensional star-cylinder,
and $\dot{t}\in\dot{C}$.
Then
\begin{eqnarray*}
	\mbox{$\dot{\mathsf{c}}_i(\dot{t})(j)$} &=& \left\{
	\begin{array}{ll}
		\dot{t}(j)  & \mbox{\em if } \;j \neq i \\
		*           & \mbox{\em if } \;j=i
	\end{array}  
	\right.
\end{eqnarray*}
for $j\in\{1,\ldots,n\}$, and 
\begin{equation*}
\mbox{$\dot{\mathsf{c}}_{i}(\dot{t})(n+1)$} 
=
\{(j=k)\in\dot{t}(n+1) \,:\, j,k\neq i\}.
\end{equation*}
We then let
$
\dot{\mathsf{c}}_{i}(\dot{C})
=
\{\dot{\mathsf{c}}_{i}(\dot{t}\,) \,:\, \dot{t}\in \dot{C}\}.
$

\item
{\em Inner cylindrification:}
Let $\dot{C}$ be an $n$-dimensional cylinder
and $i\in\{1,\ldots,n\}$. Then
$$
\dot{\creversed}_{i}(\dot{C}) 
=
\{\dot{t}\in \dot{C} \;:\;\dot{t}(i)=\ast, \text{ and}
\;\;(i=j)\notin\dot{t}(n+1) \text{ for any } j\}.
$$
\end{enumerate}

We illustrate the positive cylindric star-algebra
with the following small example.
\begin{example}
Let 
$\dot{C}_1 = \{(a,*,*,*,*,\{(3=4)\})\}$,
$\dot{C}_2 =\{(*,b,*,*,*,\{(4=5)\})\}$,
$\dot{C}_3 =\{ (a,b,*,*,*,\{(4=5)\})\}$, and
consider 
$\dot{\creversed}_3(
(\dot{\mathsf{c}}_{1,4}(\dot{C}_{1}\capdot\dot{C} _{2}))
\cupdot
\dot{C}_{3}
). 
$
Then we have the following.
\begin{small}
\begin{eqnarray*}
\dot{C}_1\capdot\dot{C}_2 
&=&
\{(a,b,*,*,*,\{(3=4),(4=5)\})\}  \\
\dot{\mathsf{c}}_{1,4}(\dot{C}_1\capdot\dot{C}_2) 
&=& 
\{(*,b,*,*,*,\{(3=5)\})\} \\
(\dot{\mathsf{c}}_{1,4}(\dot{C}_1\capdot\dot{C_2}))
\cupdot
C_{3}
&=&
\{(*,b,*,*,*,\{(3=5)\}) \\
&&
\;\;\,(a,b,*,*,*,\{(4=5)\})\} \\
\dot{\creversed}_3(
(\mathsf{c}_{1,4}(\dot{C}_1\capdot\dot{C_2}))
\cupdot
C_3)
&=&
\{
(a,b,*,*,*,\{(4=5)\})
\}
\end{eqnarray*}
\end{small}
\end{example}

Next we show that the cylindric star-algebra
has the promised property.
\begin{theorem}\label{ca2castar}
Let  $\dot{C}$ and $\dot{D}$
be $n$-dimensional star-cylinders
and $\dot{d}_{ij}$ an
$n$-dimensional star-diagonal. 
Then the following statements hold.
\begin{enumerate}
\item
$
\lsem
\dot{d}_{ij}
\rsem
=
d_{ij}.
$

\item
$\lsem\dot{C}\cupdot\dot{D}\rsem =
\lsem\dot{C}\rsem 
\,\bigcup\;\;
\lsem\dot{D}\rsem
$.

\item
$\lsem\dot{C}\capdot\dot{D}\rsem =
\lsem\dot{C}\rsem 
\,\bigcap\;\;
\lsem\dot{D}\rsem
$.

\item
$
\lsem
\dot{\mathsf{c}}_i(\dot{C})
\rsem
=\;
\mathsf{c}_i(\lsem\dot{C}\rsem),
$

\item
$
\lsem
\dot{\mathsf{\creversed}}_i(\dot{C})
\rsem
=\;
\mathsf{\creversed}_i(\lsem\dot{C}\rsem),
$
\end{enumerate}
\end{theorem}

\medskip

\begin{proof}
$\;$

\begin{enumerate}
\item
$t\in \lsem\dot{d}_{ij}\rsem$
iff
$t\preceq (\ast,\ldots,\ast,(i=j))$
iff 
$t  \in  \{t\in\mathbb{D}^n \;:\; t(i)=t(j)\}$
iff
$t\in{d}_{ij}$.

\item
$t \in \lsem\dot{C}\cupdot\dot{D}\rsem$ 
iff 
$\exists\dot{u}\in\dot{C} :
t\preceq\dot{u}$ 
or
$\exists\dot{v}\in\dot{D} : 
t\preceq\dot{v}$
iff
$t\in \lsem\dot{C}\rsem$
or 
$t\in \lsem\dot{D}\rsem$
iff
$t \in \lsem\dot{C}\rsem 
\,\bigcup\;\;
\lsem\dot{D}\rsem$.

\item
Let $t \in \lsem\dot{C}\capdot\dot{D}\rsem$. 
Then there is a star-tuple
$\dot{t}\in\dot{C}\capdot\dot{D}$ such that
$t\preceq\dot{t}$,
which again means that
there are star-tuples
$\dot{u}\in\dot{C}$ and
$\dot{v}\in\dot{D}$, such that
$\dot{t} = \dot{u}\curlywedge\dot{v}$.
As a consequence 
$t\preceq\dot{u}$
and 
$t\preceq\dot{v}$,
which implies
$t\in \lsem\dot{C}\rsem$
and 
$t\in \lsem\dot{D}\rsem$,
that is,
$t \in \lsem\dot{C}\rsem 
\,\bigcap\;\;
\lsem\dot{D}\rsem$.
The proof for the other direction is similar.


\item
Let
$t\in\lsem\dot{\mathsf{c}}_i(\dot{C})\rsem$.
Then
there is a star-tuple
$\dot{t}\in\dot{\mathsf{c}}_i(\dot{C})$
such that
$t\preceq\dot{t}$.
This in turn means that there is a star-tuple
$\dot{u}\in\dot{C}$ such that
either 
$\dot{u}=\dot{t}(i/a)$ for some $a\in\mathbb{D}$,
or 
$\dot{u}(i)=*$ and 
$\dot{u}=\dot{t}$, except
possibly
$\dot{u}(n+1)\models\theta$ where
$\theta$ is a set of equalities
involving column~$i$,
and $\dot{t}(n+1)$ does not have any conditions on $i$.

\medskip
\noindent
{\em Case 1.}
$\dot{u}=\dot{t}(i/a)$ for some $a\in\mathbb{D}$.
Then
$\dot{t}(i/a)\in\dot{C}$
which means that there is a tuple
$u\in\lsem\dot{C}\rsem$ such that 
$u\preceq\dot{t}(i/a)$.
Since
$\lsem\dot{C}\rsem\subseteq\; 
\mathsf{c}_i(\lsem\dot{C}\rsem)$,
it follows that
$u\in\mathsf{c}_i(\lsem\dot{C}\rsem)$.
Suppose $u\neq t$. 
Then $u(j)\neq t(j)$ for some
$j\in\{1\ldots,n\}$.

If $j=i$,
then
$t=u(j/t(j))\in\mathsf{c}_j(\lsem\dot{C}\rsem)
=
\mathsf{c}_i(\lsem\dot{C}\rsem)$.

If $j\neq i$ and $\dot{t}(j)=*$
it means that $\dot{u}(j)=*$,
and thus $t=u(j/t(j))\in\lsem\dot{C}\rsem$,
which in turn entails that
$t\in\mathsf{c}_i(\lsem\dot{C}\rsem)$.
Otherwise, if
$\dot{t}(j)\neq *$, then $\dot{t}(j)\in\mathbb{D}$,
which means that $\dot{u}(j)\in\mathbb{D}$, and
$u(j) = t(j)$
after~all.

\medskip
\noindent
{\em Case 2.}
$\dot{u}(i)=*$ and (possibly) 
$\dot{u}(n+1)$ contains a set of equalities
say $\theta$,
involving column~$i$,
and $\dot{t}(n+1)$ does not have any conditions on $i$.

Suppose first that $t\models\theta$.
Then $t\preceq\dot{u}$, and consequently
$t\in\lsem\dot{C}\rsem
\subseteq\,
\mathsf{c}_i(\lsem\dot{C}\rsem)$.

Suppose then that $t\not\models\theta$.
If $t$ violates an equality $(i=j)\in\theta$
it must be that $\dot{t}(j)=\dot{u}(j)=*$,
and $\dot{t}$ and $\dot{u}$ have the same
conditions on column $j$.
Let $u$ be a tuple
such that $u\preceq\dot{u}$.
Then $t(i/u(i))\in\lsem\dot{C}\rsem$,
and hence
$t\in\mathsf{c}_i(\lsem\dot{C}\rsem)$.

\medskip
For the other direction,
let 
$t\in\mathsf{c}_i(\lsem\dot{C}\rsem)$.
Then there is a tuple $u\in\lsem\dot{C}\rsem$,
such that $t(i/u(i))=u$.
Hence there is a star-tuple 
$\dot{u}\in\dot{C}$, such that $u\preceq\dot{u}$
and $t(i/u(i))\preceq\dot{u}$.
If $t\not\preceq\dot{u}$
it is because $t(i)$ violates some condition in $\dot{u}(n+1)$.
Since all conditions involving column $i$ are deleted in
$\dot{\mathsf{c}}_i(\dot{C})$, it follows that
$\dot{\mathsf{c}}_i(\dot{C})$ must contain a star-tuple $\dot{v}$
obtained by outer cylindrification of $\dot{u}$.
Then clearly
$t\preceq \dot{v}$ and
$t\models\dot{v}(n+1)$.
Consequently
$t\in \lsem\dot{\mathsf{c}}_i(\dot{C})\rsem$.
\end{enumerate}
\end{proof}

\medskip

In order to show the equivalence of positive cylindric
star-algebra and positive cylindric set algebra
we need the concept of algebra expressions.

\begin{definition}
Let $\dot{\mathbf{C}} = 
(\dot{C}_1,\ldots,\dot{C}_m,\dot{d}_{ij})_{i,j}$
be a sequence of $n$-dimensional star-cylinders and
star-diagonals. We define the set of
{\em positive cylindric star algebra expressions}
$\mbox{SCA}^{+}_n$
and values of expressions
as in Definition~\ref{caexpr},
noting that 
$\dot{\mathsf{C}}_p(\dot{\mathbf{C}}) = \dot{C}_p$,
and $\dot{\mathsf{d}}_{ij}(\dot{\mathbf{C}}) = \dot{d}_{ij}$.
\end{definition}

We now have from Theorem \ref{ca2castar}

\begin{corollary}\label{corocacastar}
For every $\mbox{SCA}^+_n$-expression 
$\dot{E}$
and the corresponding 
$\mbox{CA}^+_n$ expression $E$,
it holds that
$$
\lsem\dot{E}(\dot{\mathbf{C}})\rsem
\;=\;
E(\lsem\dot{\mathbf{C}}\rsem)
$$
for every sequence of $n$-dimensional star-cylinders 
and star-diagonals
$\dot{\mathbf{C}}$.
\end{corollary}

\subsection{Adding negation}\label{adding-neg}

From here on we also allow conditions
of the form $(i\neq j)$, $(i\neq a)$,
for $a\in\mathbb{D}$ in star-cylinders,
which then will be called
{\em extended star-cylinders}.
Conditions of the form
$(i=j),(i\neq j)$ or $(i\neq a)$
will be called {\em literals},
usually denoted~$\ell$.
In other words, in an extended
$n$-dimensional star-cylinder
each (extended) star-tuple $\dot{t}$
has a (finite) set
of literals in position $n+1$.

Everything
else remains unchanged, except
for the inner cylindrification
will be redefined below,
along with the definition
of the complement operator.
It is an easy exercise to verify
that the proofs of parts 1 -- 4 of  Theorem \ref{ca2castar}
remain valid in the presence of literals.
Complement and inner cylindrification
will be defined below.

\begin{example}\label{ex-neg-inf2}
In Example~\ref{ex-neg-inf}
we were interested in the 
negative information
as well as positive information.
The instance from Example~\ref{ex-neg-inf}
can be formally represented
as the extended star-cylinder below.

\begin{center}
\begin{tabular}{lll}
$H^-$  &             &                              \\ \hline
Alice & Volleyball  & $\{\true\}$                       \\
Bob   & $\ast$      & $\{(2\neq \mbox{ Basketball})\}$  \\
Chris & $\ast$      & $\{\true\}$
\end{tabular}
\end{center}
\end{example}

For complement and inner cylindrification
we first introduce the notion of a {\em sieve}-cylinder.
\begin{definition}\label{sieve}
Let $\dot{\mathbf{C}}$ be a sequence of
$n$-dimensional extended star-cylinders and $A$ be the
set of constants appearing therein.
For $t\in(A\cup\{*\})^n$, define
$S_{t} = \{i : t(i)=*\}$ and
$SS_{t} = \{(i,j) : t(i)=t(j)=*\}$.
For each tuple $t\in(A\cup\{*\})^n$
and each subset $SS_t^+$ of $SS_t$,
form the star-tuple $\dot{t}$
with $\dot{t}(i)=t(i)$ for $i\in\{1,\ldots,n\}$,
and $\dot{t}(n+1) =$

$$
\bigcup_{i\in S_t} \{(i\neq a) : a\in A\}
\bigcup_{(i,j)\,\in SS_{t}^{+}} \{(i=j)\}
\bigcup_{(i,j)\,\in SS_t\setminus SS_{t}^{+}}\!\!\!\!\!\! \{(i\neq j)\}.
$$
$\dot{A}$ is the extended star-cylinder
of all such star-tuples $\dot{t}$,
and it is called the {\em sieve} of 
$\dot{\mathbf{C}}$.
\end{definition}
The sieve $\dot{A}$ has the following
useful properties.

\begin{lemma}
Let $\dot{C}$ be an $n$-dimensional star-cylinder and
$\dot{A} = \{\dot{t}_1,\ldots,\dot{t}_m\}$ its sieve.
Then 
\begin{enumerate}
\item
$\lsem\dot{A}\rsem = \mathbb{D}^n$ and
$\{\lsem\{\dot{t}_1\}\rsem,\ldots,\lsem\{\dot{t}_m\}\rsem\}$
is a partition of $\lsem\dot{A}\rsem$.
\item
If $\dot{t}\wedge\dot{u}
\in\dot{C}\capdot\dot{A}$
and $\dot{t}\wedge\dot{u}\neq\dot{t}_{\emptyset}$, then
$\dot{t}\wedge\dot{u}=\dot{u}$.
\end{enumerate}
\end{lemma}

\begin{proof}
To see that
$\lsem\dot{A}\rsem = \mathbb{D}^n$,
let $t$ be an arbitrary tuple in $\mathbb{D}^n$.
By construction, there are star-tuples
$\dot{t}\in\dot{A}$ such that
$\dot{t}(i)=t(i)$ if $t(i)\in A$, and
$\dot{t}(i)=*$ if $t(i)\in\mathbb{D}\setminus A$.
Since there is the subset
$SS^+_t = \{(i,j) : t(i)=t(j), \mbox{ and } t(i)\in\mathbb{D}\setminus A\}$
we see that for one of these $\dot{t}$-tuples it holds
that $t\preceq\dot{t}$.
The fact that 
$\lsem\{\dot{t}_i\}\rsem\cap\lsem\{\dot{t}_j\}\rsem=\emptyset$
whenever $i\neq j$ follows from the fact that
if there were a tuple $t$ in the intersection,
it would have to agree with $\dot{t}_i$ and $\dot{t}_j$
on all columns with values in $A$. But the 
$SS_t^+$ set used for $\dot{t}_i$ would be different
than the one used for $\dot{t}_j$, which means
that we cannot have both 
$t\preceq\dot{t}_i$ and $t\preceq\dot{t}_j$.

\medskip
For part 2, let
$\dot{t}\wedge\dot{u}
\in\dot{C}\capdot\dot{A}$
and $\dot{t}\wedge\dot{u}\neq\dot{t}_{\emptyset}$.
We claim that $\dot{u}\preceq\dot{t}$,
which would imply
$\dot{t}\curlywedge\dot{u}~=~\dot{u}$.

Since
$\dot{t}\wedge\dot{u}\neq\dot{t}_{\emptyset}$
there is a tuple 
$
t\in
\lsem\{\dot{t}\wedge\dot{u}\}\rsem
$.
For each $i\in\{1,\ldots,n\}$,
consider $\dot{u}(i)$.
If $\dot{u}(i)=a\in A$, then 
$t(i)=a$, which means that $\dot{t}(i)=a$
or $\dot{t}(i)=*$. Consequently $\dot{u}(i)\preceq\dot{t}(i)$.
If $\dot{u}(i)=*$ 
then $t(i)\in\mathbb{D}\setminus A$,
since $(i\neq a)\in\dot{u}(n+1)$
for all $a\in\mathbb{D}\setminus A$.
Since $t(i)\preceq\dot{t}(i)$,
and $\dot{t}(i)\in A\cup\{*\}$,
it follows that $\dot{t}(i)=*$.

Then let $(i=j)\in\dot{t}(n+1)$.
Since $\dot{t}\curlywedge\dot{u}\,(n+1)$
is satisfiable, and $\dot{u}(n+1)$
contains either $(i=j)$ or $(i\neq j)$,
it follows that $(i=j)\in\dot{u}(n+1)$.
We have now shown that
$\dot{u}\preceq\dot{t}$.
\end{proof}

\begin{lemma}\label{star-cyl-dom}
Let $\dot{C}$ and $\dot{D}$ be $n$-dimensional extended
star-cylinders and $\dot{A}$ their (common) sieve.
Then
$$
\lsem\dot{C}\rsem \,\subseteq\, \lsem\dot{D}\rsem
\;\mbox{ iff }\;\,
\dot{C} \capdot \dot{A} \,\preceq\, \dot{D} \capdot \dot{A}.
$$
\end{lemma}

\begin{proof}
For the {\em if}-direction,
let $t\in\lsem\dot{C}\rsem
=
\lsem\dot{C}\capdot\dot{A}\rsem$.
Then there is a star-tuple
$\dot{t}\in\dot{C}\capdot\dot{A}$,
such that
$t\preceq\dot{t}$.
Since
$\dot{C} \capdot \dot{A} \,\preceq\, \dot{D} \capdot \dot{A}$
there is a star tuple
$\dot{u}\in\dot{D}\capdot\dot{A}$
such that 
$\dot{t}\preceq\dot{u}$.
Thus 
$t\in\lsem\dot{D}\capdot\dot{A}\rsem
=
\lsem\dot{D}\rsem$.

For the only-if direction,
let $\dot{t}_1\curlywedge\dot{u}_1\in\dot{C}\capdot\dot{A}$,
and $t\preceq\dot{t}_1$ and $t\preceq\dot{u}_1$.
Then 
$t\in\lsem\dot{C}\rsem
\subseteq\lsem\dot{D}\rsem
=
\lsem\dot{D}\capdot\dot{A}\rsem$,
so there are star-tuples $\dot{t}_2\in\dot{D}$
and $\dot{u}_2\in\dot{A}$
such that $t\preceq\dot{t}_2$
and $t\preceq\dot{u}_2$.
From part 1 of this lemma
it follows that
$\dot{u}_1=\dot{u}_2$,
and thus
$\dot{t}_1\curlywedge\dot{u}_1
=
\dot{u}_1=\dot{u}_2=\dot{t}_2\curlywedge\dot{u}_2$.
Consequently
$\dot{t}_1\curlywedge\dot{u}_1
\preceq
\dot{t}_2\curlywedge\dot{u}_2$.
\end{proof}

We can now define the desired operations.
\begin{definition}\label{dotneg}
Let $\dot{A}$ be the sieve of
$\dot{\mathbf{C}}$ and 
$\dot{C}$ be an extended star-cylinder 
in~$\dot{\mathbf{C}}$.
Then
\begin{enumerate}
\item
$
\boldsymbol{\dot{\neg}}\,\dot{C}
=
\{
\dot{t}\in\dot{A} \,:\,
\{\dot{t}\}\capdot\dot{C}=\{\dot{t}_{\emptyset}\}
\}
$.
and

\item
$
\hat{\creversed}_i(\dot{C})
=
\{\dot{t}\in\dot{C}\capdot\dot{A} \,:\,
(\dot{\mathsf{c}}_i(\{\dot{t}\,\}) \capdot \dot{A})
\preceq\,
(\dot{C} \capdot \dot{A})\}.
$
\end{enumerate}
\end{definition}

\begin{example}\label{exmp-neg=cyl}
Let $\dot{C} = \{(a,*,\{\true\})\}$.
Then $\dot{A}$ is shown in the
extended star-cylinder below,
and 
$\boldsymbol{\dot{\neg}}\,\dot{C}$
consists of the first, second,
and fourth tuples of $\dot{A}$.

\begin{center}
\begin{tabular}{ccl}
$\dot{A}$ &&\\
\hline
$\ast$ & $\ast$ & $\{(1\neq a),(2\neq a),(1=2)\}$ \\
$\ast$ & $\ast$ & $\{(1\neq a),(2\neq a),(1\neq 2)\}$ \\
$a$    & $\ast$ & $\{(2\neq a)\}$ \\
$\ast$ & $a$    & $\{(1\neq a)\}$ \\
$a$    & $a$    & $\{\true\}$
\end{tabular}
\end{center}
Now, let $\dot{C} =
\{(a,*,\{(2\neq a)\}),
(a,a,\{\true\})\}$.
Then $\dot{A}$ is as above,
and
$\hat{\creversed}_2(\dot{C}) = \dot{C}$
as the reader easily can verify.
\end{example}
	
We can now verify that the new operators work as expected.

\begin{theorem}\label{prop-neg-def}
Let $\dot{C}$ be an extended star-cylinder.
Then
\begin{enumerate}

\item
$
\lsem\boldsymbol{\dot{\neg}}\,\dot{C}\rsem
\;=\;
\overline{\lsem\dot{C}\rsem}
$

\item
$
\lsem
\hat{\creversed}_i(\dot{C})
\rsem
\;=\;
\creversed_i(\lsem\dot{C}\rsem).
$
\end{enumerate}
\end{theorem}

\begin{proof}
For part 1, 
it is easy to see
that $\lsem\boldsymbol{\dot{\neg}}\,\dot{C}\rsem
\cap \lsem\dot{C}\rsem = \emptyset$
which implies  $\lsem\boldsymbol{\dot{\neg}}\,\dot{C}\rsem
\subseteq
\overline{\lsem\dot{C}\rsem} $.
For a proof of the other direction of part 1,
for each tuple
$t\in\overline{\lsem\dot{C}\rsem}$,
we construct the star-tuple $\dot{t}$,
where
$\dot{t}(i)=t(i)$ if $t(i) \in A$,
and $\dot{t}(i)=\ast$ if $t(i) \not\in A$.
We then choose a subset
$SS_t^+$ of $SS_t$
where $(i,j)\in SS^+$
if and only if
$t(i)=t(j)$.
We insert in
$\dot{t}(n+1)$
the condition
$(i=j)$ for each $(i,j)\in SS_t^+$,
and $(i\neq j)$ for each $(i,j)\in SS_t\setminus SS_t^+$,
Then clearly
$t\in\lsem\{\dot{t}\}\rsem$ and
$\dot{t}\in\dot{A}$.
It remains to show that
$\dot{t}\in\boldsymbol{\dot{\neg}}\, \dot{C}$.
Towards a contradiction,
suppose that there is a star-tuple
$\dot{u} \in \dot{C}$ such that
$\dot{t} \curlywedge \dot{u} \neq \dot{t}_{\emptyset}$.
In other words,
$\dot{t}(n+1)\cup\dot{u}(n+1)$
is satisfiable.
Thus, whenever $\dot{t}(i)\in\mathbb{D}$,
we must have $\dot{u}(i)=\dot{t}(i)=t(i)\in A$.
Furthermore,
for each $(i,j)\in SS_t$
there is a literal involving
$i$ and $j$ in $\dot{t}(n+1)$.
Therefore $\dot{u}(n+1)$ can
consist of only a subset of these literals.
It follows that 
$t\preceq\dot{t}\preceq\dot{u}\in\dot{C}$,
meaning that $t\in\lsem\dot{C}\rsem$,
contradicting our initial assumption.

\bigskip
For a proof of part 2 of the theorem,
let $t \in \lsem\hat{\creversed}_i(\dot{C})\rsem$.
Then $t \in \lsem\{\dot{t}\in\dot{A} \,:\,
(\dot{\mathsf{c}}_i(\{\dot{t}\,\}) \capdot \dot{A})
\preceq\,
(\dot{C} \capdot \dot{A})\}\rsem$.
Therefore there is a star tuple
$\dot{t} \in \dot{A}$ such that,
$t \preceq \dot{t}$ and
$(\dot{\mathsf{c}}_i(\{\dot{t}\,\}) \capdot \dot{A})
\preceq\,
(\dot{C} \capdot \dot{A})$.
{\color{red}Lemma~\ref{star-cyl-dom}} then gives us
$\lsem \dot{\mathsf{c}}_i(\{\dot{t}\,\}) \rsem 
\subseteq\lsem 
\dot{C}
\rsem$,
and Theorem~\ref{ca2castar}, part 4
(which still holds for extended star-cylinders)
tell us that
$\lsem \dot{\mathsf{c}}_i(\{\dot{t}\,\}) \rsem
\;=\;
{\mathsf{c}}_i( \lsem \{\dot{t}\,\} \rsem)$
which implies
$\lsem \dot{\mathsf{c}}_i(\{\dot{t}\,\}) \rsem 
\subseteq\lsem 
\dot{C}
\rsem$.
By the definition 
of inner cylindrification in CA,
the last containment implies that
$\lsem \{\dot{t}\,\} \rsem 
\;\subseteq\;
\creversed_i(\lsem\dot{C}\rsem)$.
Consequently
 $t \in \creversed_i(\lsem\dot{C}\rsem)$.
	
For the other direction, let
$t \in \creversed_i(\lsem\dot{C}\rsem)$,
which implies
$\mathsf{c}_i(\{t\})
\subseteq\lsem\dot{C}\rsem$.
Then there is a star-tuple
 $\dot{t} \in \dot{C}\capdot\dot{A}$,
 such that
$\mathsf{c}_i(\{t\})
\subseteq
\mathsf{c}_i(\lsem\{\dot{t}\}\rsem)
\subseteq
\lsem\dot{C}\rsem$.
Consequently,
$\lsem \dot{\mathsf{c}}_i(\dot{t})\rsem
\subseteq
\lsem\dot{C}\rsem$,
which by Lemma~\ref{star-cyl-dom}
proves that
$(\dot{\mathsf{c}}_i(\{\dot{t}\,\}) \capdot \dot{A})
\preceq\,
(\dot{C} \capdot \dot{A})$.
Moreover, the first part of
Lemma~\ref{star-cyl-dom} implies that
$t \in \lsem\{\dot{t}\in\dot{A} \,:\,
(\dot{\mathsf{c}}_i(\{\dot{t}\,\}) \capdot \dot{A})
\preceq\,
(\dot{C} \capdot \dot{A})\}\rsem$.
\end{proof}
	
We can thus conclude
\begin{corollary}\label{corocacastar2}
For every $\mbox{SCA}_n$-expression
$\dot{E}$
and the corresponding
$\mbox{CA}_n$-expression $E$,
it holds that
$$
\lsem\dot{E}(\dot{\mathbf{C}})\rsem
\;=\;
E(\lsem\dot{\mathbf{C}}\rsem)
$$
for every sequence of $n$-dimensional
extended star-cylinders
and star-diagonals~$\dot{\mathbf{C}}$.
\end{corollary}

%

\section{Stored databases with\\ universal nulls (u-databases)}\label{stored}

We now show how to use the cylindric star-algebra
to evaluate FO-queries on stored databases
containing universal nulls.
Let $k$ be a positive integer.
Then a $k$-ary star-relation 
$\dot{R}$ is a finite set
of star-tuples of arity $k$.
In other words, a $k$-ary star-relation
is a star-cylinder of dimension $k$.
A sequence $\dot{\mathbf{R}}$
of star-relations 
(over schema $\mathbf{R}$)
is called a {\em stored database}.
Examples \ref{ex1} 
and \ref{ex-neg-inf2}
show stored databases.
Everything that is defined for
star-cylinders applies to
$k$-ary star-relations.
The exception is that no
operators from the cylindric star-algebra
are applied to star-relations.
To do that, we first need to expand the
stored database $\dot{\mathbf{R}}$.
\begin{definition}\label{expansion}
Let $\dot{t}$ be a $k$-ary star-tuple,
and  $n\geq k$.
Then
$
\dot{\mathsf{h}}^n(\dot{t})
$,
the \textit{$n$-expansion} of $\dot{t}$,
is the $n$-ary star-tuple~$\dot{u}$,
where
\vspace{-.2cm}
$$
\dot{u}(i) = \left\{
             \begin{array}{ll}
             \dot{t}(i)   & \mbox{\em if } i\in\{1,\ldots,k\}   \\
             *            & \mbox{\em if } i\in\{k+1,\ldots,n\} \\
             \dot{t}(k+1) & \mbox{\em if } i=n+1,
             \end{array}
             \right.   
$$

For a stored relation $\dot{R}$
and stored database $\dot{\mathbf{R}}$
we have 
\vspace{-.2cm}
\begin{eqnarray*}
\dot{\mathsf{h}}^n(\dot{R})
& = &
\{\dot{\mathsf{h}}^n(\dot{t})
\;:\;
\dot{t}\in\dot{R}\} \\
\dot{\mathsf{h}}^n(\dot{\mathbf{R}})
& = &
(\dot{\mathsf{h}}^n(\dot{R}_1),
\ldots,
\dot{\mathsf{h}}^n(\dot{R}_m),
\dot{\mathsf{h}}^n(\{(a,a) : a\in\mathbb{D}\})).
\end{eqnarray*}
\end{definition}
In other words,
$\dot{\mathsf{h}}^n(\dot{\mathbf{R}})$
is the sequence of star-cylinders
obtained by 
moving the conditions in column $k+1$
to column $n+1$, and
filling columns $k+1,\ldots,n$
with "*"'s in each $k$-ary relation.
Examples \ref{ex1} and \ref{ex2}
illustrate the expansion of star-relations.

On the other hand, 
a $k$-ary star-relation $\dot{R}$ can also be viewed
as a finite representative of the infinite
relation 
$\lsem\dot{R}\rsem 
= 
\{t\in\mathbb{D}^k :\, t\preceq\dot{t} \mbox{ for some } \dot{t}\in\dot{R}\}$,
and the stored database $\dot{\mathbf{R}}$
a finite representative of the infinite instance
$I(\dot{\mathbf{R}})$,
as in the following definition.
\begin{definition}\label{instance}
Let 
$
\dot{\mathbf{R}} 
=
(\dot{R}_1,\ldots,\dot{R}_m)
$ 
be a stored database.
Then the
instance defined by
$\dot{\mathbf{R}}$
is
\vspace{-.3cm}
$$
I(\dot{\mathbf{R}})
=
(\lsem\dot{R}_1\rsem,
\ldots,
\lsem\dot{R}_m\rsem,
\{(a,a) : a\in\mathbb{D}\}).
$$
\end{definition}
The instance and expansion of $\dot{\mathbf{R}}$
are related as follows.
\begin{lemma}\label{lemma2}
$\lsem\dot{\mathsf{h}}^n(\dot{\mathbf{R}})\rsem
\;=\;
\mathsf{h}^n(I(\dot{\mathbf{R}}))$.
\end{lemma}\vspace{-.15cm}
\begin{proof}
	It directly follows from the definition of 
	$\mathsf{h}^n$, $\dot{\mathsf{h}}^n$ and $\lsem \rsem$.
\end{proof}

\newpage
We are now ready for our main result.
\begin{theorem}\label{main0}
For every FO$_n$-formula $\varphi$
there is an (extended) SCA$_n$ expression 
$\dot{E}_{\varphi}$,
such that for every stored 
database~$\dot{\mathbf{R}}$
$$
\mathsf{h}^n(\varphi^{I(\dot{\mathbf{R}})})
\;=\;
\lsem\dot{E}_{\varphi}(\dot{\mathsf{h}}^n(\dot{\mathbf{R}}))\rsem.
$$
\end{theorem}

\noindent
{\bf Proof:}
$
\mathsf{h}^n(\varphi^{I(\dot{\mathbf{R}})})
\;=\; 
E_{\varphi}(\mathsf{h}^n(I(\dot{\mathbf{R}}))
\;=\; 
E_{\varphi}(\lsem\dot{\mathsf{h}}^n(\dot{\mathbf{R}})\rsem)
\;=\; 
\lsem\dot{E}_{\varphi}(\dot{\mathsf{h}}^n(\dot{\mathbf{R}}))\rsem.
$
The first equality follows from Theorem~\ref{thrm1},
the second from Lemma \ref{lemma2},
and the third from Corollaries 
\ref{corocacastar} and \ref{corocacastar2}.


\section{Adding existential nulls}\label{naive}

Let $\mathbb{N} = \{\bot_1,\bot_2,\ldots\}$
be a countable infinite set of 
{\em existential nulls}.
An instance $I$ where the relations
are over $\mathbb{D}\cup\mathbb{N}$,
is in the literature variably called a
naive table 
\cite{DBLP:journals/jacm/ImielinskiL84,
libkin-naive}
or a
generalized instance \cite{DBLP:conf/pods/FaginKPT09}.
In either case, such an instance
is taken to represent an
{\em incomplete instance},
i.e.\ a (possibly) infinite set of instances.
In this paper we follow the
model-theoretic approach of~\cite{DBLP:conf/pods/FaginKPT09}.
The elements in $\mathbb{D}$
represent known objects,
whereas elements in
$\mathbb{N}$ represent generic objects.
Each generic object could turn out to
be equal to one of the known objects,
to another generic object, or represent
an object different from all other objects. 
We extend our notation to include
$\univ(I)$, the {\em universe} of
instance $I$. So far we have assumed 
that $\univ(I)=\mathbb{D}$,
but in this section we allow instances
whose universe is any set between
$\mathbb{D}$ and
$\mathbb{D}\cup\mathbb{N}$.
We are lead to the following definitions.

\begin{definition}
Let $h$ be a mapping on 
$\mathbb{D}\cup\mathbb{N}$
that is identity on $\mathbb{D}$,
and let $I$ and $J$ be instances
(over $\mathscr{R}$),
such that $h(\univ(I)) = \univ(J)$.
We say that $h$ is a 
{\em possible world homomorphism}
from $I$ to $J$, if
$h(R^I_p)\subseteq R^J_p$ for all $p$, and
$h(\approx^I)=\; \approx^J$.
This is denoted $I\rightarrow_h J$. 
\end{definition}

\begin{definition}\label{rep}
Let $I$ be an instance with
$\mathbb{D}\subseteq\univ(I)$ $\subseteq\mathbb{D}\cup\mathbb{N}$.
Then the set of instances represented by
$I$ is
$$
\rep(I)
\;=\;
\{J \;:\; \boldsymbol{\exists}\, h\;\, s.t.\ I\rightarrow_h J\}.
$$
\end{definition}
We can now formulate the (standard) notion of a
{\em certain answer} to a query.\footnote{Here 
we of course assume that
valuations have range $\univ(J)$,
and that other details are adjusted
accordingly.}
By FO$^+$ below we mean the set of
all FO-formulas not using negation.
\begin{definition}
Let $I$ be an incomplete instance and 
$\varphi$ an FO$^+$-formula.
The certain answer to $\varphi$ on $I$ is
$$
\cert(\varphi,I)
\;\;=
\bigcap_{J\in\rep(I)} \varphi^J.
$$
\end{definition}

\noindent
{\bf Back to cylinders.}
We now extend positive 
$n$-dimen\-sional cylinders
to be subsets of $\mathbb{D}\cup\mathbb{N}$,
and use the notation
$\univ(\mathbf{C})$ and $\univ(C)$ with the
obvious meanings.
This also applies to the notation
$\mathbf{C}\rightarrow_h\mathbf{D}$,
and $\rep(\mathbf{C})$.
The operators of the positive cylindric set algebra
CA$^+$ remain the same,
except $\mathbb{D}$ is substituted with
$\univ(\mathbf{C})$ or $\univ(C)$,
i.e.\ we use naive evaluation.
For instance, the outer cylindrification now
becomes
\begin{small}
$$
\mathsf{c}_{i}(C) = \{t\in\univ(C)^n \;:\; t(i/x)\in C, 
\mbox{ for some } x\in\univ(C)\}.
$$
\end{small}
The crucial property of 
the positive cylindric star-algebra is the following.

\begin{theorem}\label{monotone2}
Let $E$ be an expression in 
CA$^+_n$,
and $\mathbf{C}$ and $\mathbf{D}$
sequences of $n$-dimensional 
naive cylinders and diagonals.
If 
$\mathbf{C}\rightarrow_h\mathbf{D}$
for some possible world homomorphism $h$, 
then
$E(\mathbf{C})\rightarrow_hE(\mathbf{D})$.
\end{theorem}

\begin{proof}
Suppose $\mathbf{C}\rightarrow_h\mathbf{D}$.
We show by induction on the structure of $E$
that 
$E(\mathbf{C})\rightarrow_hE(\mathbf{D})$.
\begin{itemize}
\item
For $E = \mathsf{C}_i$
and $E = \mathsf{d}_{ij}$
the claim follows directly from
the definition of a possible world
homomorphism.

\item
Let $t\in h(F\bigcup G\,(\mathbf{C}))
\;=\;
h(F(\mathbf{C}) \,\cup\, G(\mathbf{C}))
\;=\;
h(F(\mathbf{C})) \,\cup\; h(G(\mathbf{C}))$.
Then there is a tuple $s$ in
$F(\mathbf{C})$ or in $G(\mathbf{C})$
such that $t=h(s)$.
If $s$ is in, say, $F(\mathbf{C})$,
then,
since 
$F(\mathbf{C})\rightarrow_hF(\mathbf{D})$
and $F(\mathbf{D})\subseteq 
F\bigcup G\,(\mathbf{D})$,
it follows that
$t=h(s)\in F\bigcup G\,(\mathbf{D})$.

\item
Let $t\in h(F\bigcap G\;(\mathbf{C}))
\;=\;
h(F(\mathbf{C}) \,\cap\, G(\mathbf{C}))$.
Then there is a tuple $s$ in
$F(\mathbf{C})$ and 
a tuple $s'$ in $G(\mathbf{C})$
such that $t=h(s)=h(s')$.
Thus $h(s)\in h(F(\mathbf{C}))
\subseteq F(\mathbf{D})$,
and
$h(s')\in h(G(\mathbf{C}))
\subseteq G(\mathbf{D})$.
Consequently
$t=h(s)=h(s')\in 
F(\mathbf{D}) \,\cap\; G(\mathbf{D})
\;=\;
F\bigcap G\,(\mathbf{D})$.
	
\item
Let $t\in h(\mathsf{c}_i(F(\mathbf{C})))$.
Then there is an 
$s\in\mathsf{c}_i(F(\mathbf{C}))$,
such that $t=h(s)$.
Furthermore,
$s(i/x)\in F(\mathbf{C})$ for some 
$x\in\univ(\mathbf{C})$.
Then $h(s(i/x)) = h(s)(i/h(x))\in h(F(\mathbf{C}))$, 
for $h(x)\in h(\univ(\mathbf{C}))
=
\univ(\mathbf{D})$,
This means that 
$t=h(s)\in\mathsf{c}_i(F(\mathbf{D}))$.
	
\item
Let $t\in h(\creversed_i(F(\mathbf{C})))$.
Then there is an $s\in\creversed_i(F(\mathbf{C}))$,
such that $t=h(s)$.
Furthermore,
$s(i/x)\in F(\mathbf{C})$ for all $x\in\univ(\mathbf{C})$.
Then
$h(s(i/x)) = 
h(s)(i/h(x)) =
t(i/h(x))\in h(F(\mathbf{C}))$ 
for all 
$x\in\univ(\mathbf{C})$.
In other words,
$t(i/y)\in h(F(\mathbf{C}))
\;\subseteq\;
F(\mathbf{D})$ 
for all 
$y\in h(\univ(\mathbf{C})) = \univ(\mathbf{D})$.
Thus
$t\in\creversed_i(F(\mathbf{D}))$
\end{itemize}
\end{proof}

Also, for an $n$-dimensional naive cylinder $C$,
we denote the subset $C\cap\mathbb{D}^n$
by~$C_{\downarrow}$.
We now have our main "naive evaluation" theorem.
\begin{theorem}\label{naivemain}
Let $\mathbf{C}$ be a sequence of $n$-dimensional 
naive cylinders and diagonals, 
and let $E$ be an expression in~$\mbox{CA}^+_n$. 
Then
$$
E(\mathbf{C})_{\downarrow}
\;\;\;=
\bigcap_{\mathbf{D}\;\in\rep(\mathbf{C})}
E(\mathbf{D}).
$$
\end{theorem}
\begin{proof}
Let $t\in E(\mathbf{C})_{\downarrow}\subseteq E(\mathbf{C})$,
and 
$\mathbf{D}\in\rep(\mathbf{C})$.
Since $\mathbf{C}\rightarrow_h \mathbf{D}$
for some possible world homomorphism~$h$,
by Theorem \ref{monotone2}, 
$h(t)\in E(\mathbf{D})$.
Since $t$ is all constants,
$h(t)=t$ for all $h$.
In other words,
$t\in E(\mathbf{D})$,
for all $\mathbf{D}\in\rep(\mathbf{C})$.

For the $\supseteq$-direction,
let $t\in\bigcap_{\mathbf{D}\;\in\rep(\mathbf{C})}
E(\mathbf{D})$.
Then $t\in\mathbb{D}^n$,
and for all possible world homomorphisms
$h$ it holds that
$t\in E(h(\mathbf{C}))$.
Since identity is a valid $h$,
is follows that
$t\in E(\mathbf{C})$,
and since $t$ is all constants
we have 
$t\in E(\mathbf{C})_{\downarrow}$.
\end{proof}

\subsubsection*{Mixing existential and universal nulls}

We want to achieve a representation mechanism
able to handle both universal nulls and naive
existential nulls. To this end we need the
following definition.

\begin{definition}
A naive $n$-dimensional (positive) star-cylinder is
a finite subset $\ddot{C}$ of
$(\mathbb{D}\cup\mathbb{N}\cup\{\ast\})^n \times\, \powerset(\Theta_n)$.
A naive diagonal is defined as
$\ddot{d}_{ij} = \{(x,x) : x \in \univ(\ddot{\mathbf{C}})\}$.
A sequence of $n$-dimensional star-cylinders and diagonals
is denoted~$\ddot{\mathbf{C}}$.
\end{definition}

After this we extend Definitions 
\ref{normalform},
\ref{domin}, 
and \ref{meet} 
in Section \ref{caandcastar} from
star-cylinders to naive star-cylinders,
by replacing $\mathbb{D}$ with
$\univ(\ddot{\mathbf{C}})$ or
$\univ(\ddot{C})$ everywhere.
Theorem \ref{normalform} will still hold,
but Corollary \ref{corocacastar} only
holds in the weakened form
as Corollary \ref{naivecoinit} below.
First we need two lemmas and a definition.


\begin{lemma}\label{homegacommute}
Suppose all possible world homomorphisms $h$
are extended by letting $h(\ast) = \ast$.
Let $\ddot{C}$ be an $n$-dimensional naive star-cylinder. 
Then
$$
h(\lsem\ddot{C}\rsem)
\,=\,
\lsem h(\ddot{C}) \rsem,
$$
for all possible world homomorphisms $h$.
\end{lemma}

\begin{proof}
Let $t\in h(\lsem\ddot{C}\rsem)$.
Then there exists a tuple 
$u\in\lsem\ddot{C}\rsem$, 
such that $t=h(u)$.
Also there exists a naive star-tuple
$\ddot{u}\in\lsem\ddot{C}\rsem$,
such that $u \preceq \ddot{u}$.
Now it is sufficient to show that
$t \preceq h(\ddot{u})$,
for all $i\in \{1,2, \ldots, n\}$.

If $\ddot{u}(i) \in \mathbb{D}$,
then $u(i) =\ddot{u}(i)$.
Also, homomorphisms are identity on constants
and therefore $h(u(i))=u(i)$, 
which implies $t(i)=u(i)$.

If $\ddot{u}(i)=\ast$,
then $u(i) \in \univ(\ddot{C})$.
As a result $h(u(i))\in \univ(h(\ddot{C}))$, which implies $t(i)\preceq
\ast=h(\ddot{u}(i))$,
since homomorphisms map stars to themselves.

If $\ddot{u}(i) \in \mathbb{N}$,
then $u(i) =\ddot{u}(i)$,
which implies $t(i)=h(u(i))= h(\ddot{u}(i))$.

\medskip

For the other direction,
let $t\in\lsem h(\ddot{C})\rsem$.
Then there exists a tuple 
$\ddot{t} \in h(\ddot{C})$
and a tuple
$\ddot{u} \in \ddot{C}$,
such that
$t \preceq \ddot{t}$
and
$\ddot{t} = h(\ddot{u})$.
Consequently, 
$t \preceq h(\ddot{u})$.
We show that we can find a tuple 
$u \in \lsem\ddot{u}\rsem$ such that
$h(u)=t$.

If $\ddot{u}(i) \in \mathbb{D}$,
then $u(i) =\ddot{u}(i)$.
Since $h$ is the  identity on constants
$h(\ddot{u}(i))=\ddot{u}(i)$, 
which implies $t(i)=u(i)$.

If $\ddot{u}(i)=\ast$,
then $h(\ddot{u}(i))=\ast$.
As $h$ is onto $\univ(h(\ddot{C}))$,
it follows that there is a value $\ddot{u}(i)\in \univ(\ddot{C})$,
such that $h(u(i))=t(i)$.

If $\ddot{u}(i) \in \mathbb{N}$,
then $u(i) =\ddot{u}(i)$,
which implies $t(i)= h(\ddot{u}(i))=h(u(i))$.
\end{proof}

\begin{definition}\label{coinit}
Let $\I$ and $\J$ be sets of instances.
We say that $\I$ and $\J$ are
\textit{co-initial},
denoted $\I\sim\J$,
if for each instance
$J\in\J$ there is an instance $I\in\I$,
and a possible world homomorphism $h$,
such that $I\rightarrow_hJ$,
and vice-versa.
\end{definition}

In the context of naive star-cylinders
Corollary \ref{corocacastar} 
will be weakened as follows.
\begin{corollary}\label{naivecoinit}
For every SCA$^+_n$-expression 
$\dot{E}$
and the corresponding
CA$_n$-expression $E$,
it holds that
$$
\rep(\lsem\dot{E}(\ddot{\mathbf{C}})\rsem)
\;\sim\;
E(\rep(\lsem\ddot{\mathbf{C}}\rsem)),
$$
for every sequence of $n$-dimensional 
naive star-cylinders 
and star-diagonals
$\ddot{\mathbf{C}}$.
\end{corollary}

\begin{proof}
We have
$
\rep(\lsem\dot{E}(\ddot{\mathbf{C}})\rsem)
\;\sim\;
\rep(E(\lsem\ddot{\mathbf{C}}\rsem))
$
from Corollary~\ref{corocacastar}.
It remains to show that
$
\rep(E(\lsem\ddot{\mathbf{C}}\rsem)
\;\sim\;
E(\rep(\lsem\ddot{\mathbf{C}}\rsem)).
$
Let's denote
$\lsem\ddot{\mathbf{C}}\rsem$ by $\mathbf{C}$.
We'll show that
$
\rep(E(\mathbf{C}))
\;\sim\;
E(\rep(\mathbf{C})).$


Let $D\in E(\rep(\mathbf{C}))$,
meaning that $D=E(\mathbf{C}')$ for some
$\mathbf{C}'\in\rep(\mathbf{C})$.
Then there is a possible world homomorphism 
$h$ such that $\mathbf{C}\rightarrow_h \mathbf{C}'$.
Theorem \ref{monotone2} then yields
$E(\mathbf{C}) \rightarrow_h E(\mathbf{C}') $,
and since 
$E(\mathbf{C})\in\rep(E(\mathbf{C}))$
one direction of Definition \ref{coinit}
is satisfied. 

Then let $D\in\rep(E(\mathbf{C}))$.
Then there is a possible world homeomorphism $h$,
such that $E(\mathbf{C}) \rightarrow_h D$.
Since $E(\mathbf{C})\in E(\rep(\mathbf{C}))$,
it means that the vice-versa
direction is also satisfied.
\end{proof}

\subsubsection*{Naive evaluation of existential nulls}

We extend Definition \ref{rep} from
infinite instances to
sequences of naive star-cylinders
as follows.
\begin{definition}
Let $\ddot{\mathbf{C}}$ be a sequence
of $n$-dimensional naive star-cylinders
and diagonals with 
$\univ(\ddot{\mathbf{C}})
=
\mathbb{D}\,\cup\mathbb{N}$.
Then the (infinite) set of 
(infinite) $n$-dimensional cylinders
represented by
$\ddot{\mathbf{C}}$ is
$$
\rep(\lsem\ddot{\mathbf{C}}\rsem)
\;=\;
\{
\mathbf{D} \,:\, \lsem\ddot{\mathbf{C}}\rsem\!\rightarrow_h \mathbf{D}
\}.
$$
\end{definition}
For a naive star-cylinder $\ddot{\mathbf{C}}$
we let 
$\ddot{\mathbf{C}}_{\downarrow}
=
\ddot{\mathbf{C}} \cap (\mathbb{D}\cup\{\ast\})^n$.
We note that obviously 
$\lsem\ddot{\mathbf{C}}_{\downarrow}\rsem
\;=\;
(\lsem\ddot{\mathbf{C}}\rsem)_{\downarrow}$,
and that if
$\rep(\ddot{\mathbf{C}}) \sim \rep(\ddot{\mathbf{D}})$,
then
$\ddot{\mathbf{C}}_{\downarrow} 
\;=\;
\ddot{\mathbf{D}}_{\downarrow}$.
We now have the main result of this
section.
\begin{theorem}\label{main2}
For every SCA$^+$-expression 
$\dot{E}$
and the corresponding
CA$^+$-expression $E$,
it holds that
$$
\lsem\dot{E}(\ddot{\mathbf{C}})_{\downarrow}\rsem
\;\;=\!
\bigcap_{\mathbf{D}\in\rep(\lsem\ddot{\mathbf{C}}\rsem)}
\!E(\mathbf{D}).
$$
for every sequence $\ddot{\mathbf{C}}$ of naive star-cylinders and diagonals.
\end{theorem}
\begin{proof}
$
\lsem\dot{E}(\ddot{\mathbf{C}})_{\downarrow}\rsem
\;=\;
\lsem\dot{E}(\ddot{\mathbf{C}})\rsem_{\downarrow}
\;=\;
(E(\lsem\ddot{\mathbf{C}}\rsem))_{\downarrow}
\;=\;
\bigcap_{\;\mathbf{C}\in\rep(\lsem\ddot{\mathbf{C}}\rsem)}
E(\mathbf{C}).
$
The second equality follows from
Corollary \ref{naivecoinit},
the third from 
Theorem \ref{naivemain}.
\end{proof}

\subsubsection*{Stored databases with universal 
and existential nulls (ue-databases)}

We extend the Definitions 
\ref{expansion} 
and 
\ref{instance} 
of Section \ref{stored} from
stored databases to naive stored databases
(ue-databases)
by substituting $\mathbb{D}$ with
$\mathbb{D}\cup\mathbb{N}$ everywhere.
Lemma \ref{lemma2}
then becomes
\begin{lemma}\label{lemma3}
Let $\ddot{\mathbf{C}}$ be a stored
ue-database with universe
$\mathbb{D}\cup\mathbb{N}$.
Then
$\lsem\dot{\mathsf{h}}^n(\ddot{\mathbf{R}})\rsem
\;=\;
\mathsf{h}^n(I(\ddot{\mathbf{R}}))$.
\end{lemma}

We first note that Theorem \ref{main0}
in the ue-setting becomes
\begin{theorem}\label{nine}
For every FO$^+_n$-formula $\varphi$
there is an SCA$^{+}_n$ expression
$\dot{E}_{\varphi}$,
such that 
$$
\lsem\dot{E}_{\varphi}(\dot{\mathsf{h}}^n(\ddot{\mathbf{R}}))\rsem
\;=\;
\mathsf{h}^n(\varphi^{I(\ddot{\mathbf{R}})})
$$
for every stored
ue-database~$\ddot{\mathbf{R}}$
\end{theorem}

We also have
\begin{theorem}\label{sim}
For every FO$^+_n$-formula $\varphi$
there is a CA$^{+}_n$ expression
$\dot{E}_{\varphi}$,
such that 
$$
\rep(\lsem\dot{E}_{\varphi}(\dot{\mathsf{h}}^n(\ddot{\mathbf{R}}))\rsem)
\;\sim\;
\{\mathsf{h}^n(\varphi^{J}) \,:\, J\in\rep(\lsem\ddot{\mathbf{R}}\rsem)\}
$$
for every stored
ue-database~$\ddot{\mathbf{R}}$
\end{theorem}

We have now arrived our main theorem for ue-databases.

\begin{theorem}
For every FO$^+_n$-formula $\varphi$
there is an SCA$^{+}_n$ expression
$\dot{E}_{\varphi}$,
such that
$$
\lsem\dot{E}_{\varphi}(\dot{\mathsf{h}}^n(\ddot{\mathbf{R}}))_{\downarrow}\rsem
\;=\;
\bigcap_{J\in\rep(\lsem\ddot{\mathbf{R}}\rsem)}
\mathsf{h}^n(\varphi^{J})
$$
for every stored
ue-database~$\ddot{\mathbf{R}}$
\end{theorem}
\begin{proof}
We have
$
\{\mathsf{h}^n(\varphi^{J}) \,:\, J\in\rep(\lsem\ddot{\mathbf{R}}\rsem)\}
\;\sim\;$
$\rep(\lsem\dot{E}_{\varphi}(\dot{\mathsf{h}}^n(\ddot{\mathbf{R}}))\rsem)
$
by Theorem \ref{sim}.
Hence

\noindent
$
\bigcap_{J\in\rep(\lsem\ddot{\mathbf{R}}\rsem)}
\mathsf{h}^n(\varphi^{J})
\;=\;
\bigcap\;\rep(\lsem\dot{E}_{\varphi}(\dot{\mathsf{h}}^n(\ddot{\mathbf{R}}))\rsem)
\;=\;
$\\
$
(\lsem\dot{E}_{\varphi}(\ddot{\mathsf{h}}^n(\ddot{\mathbf{R}}))\rsem)_{\downarrow}
\;=\;
\lsem\dot{E}_{\varphi}(\dot{\mathsf{h}}^n(\ddot{\mathbf{R}}))_{\downarrow}\rsem
$.
\end{proof}

\section{Complexity}\label{complexity}
In this section we
provide complexity results
for Cylindric Star Algebra and
Star Cylinders.
We start by defining the size of 
extended
star-cylinders.
Let $\dot{\mathbf{C}}$ be a sequence of $n$-dimensional 
extended star-cylinders and diagonals.
By $|\dot{\mathbf{C}}|$ we
denote the larger of 
the number of star-tuples in 
$\dot{\mathbf{C}}$ and
the number of literals in the
star-tuple with the largest condition
column $n+1$. The same notation also
applies to sequences of naive 
star-cylinders~$\ddot{\mathbf{C}}$.

First, we investigate the complexity of
evaluating SCA$_{•}$-expressions
over naive star-cylinders
and then we characterize various membership
and containment problems.
It turns out $\dot{E}(\dot{\mathbf{C}})$
can be computed efficiently for 
$\mbox{SCA}_n$-expressions~$\dot{E}$,
even though universal quantification
and negation are allowed.

\begin{theorem}\label{CA-complexity}
Let $\dot{E}$ be a fixed $\mbox{SCA}_n$-expression,
and $\dot{\mathbf{C}}$
a sequence of $n$-dimen\-sional 
extended star-cylinders and diagonals.
Then there is a polynomial $\pi$,
such that
$|\dot{E}(\dot{\mathbf{C}})|
= 
\mathcal{O}(\pi(|\dot{\mathbf{C}}|))$.
Moreover, 
$\dot{E}(\dot{\mathbf{C}})$
can be computed in time 
$\mathcal{O}(\pi(|\dot{\mathbf{C}}|))$,
and if negation is not used in $\dot{E}$
this applies to naive star-cylinders 
$\ddot{\mathbf{C}}$
as well.
\end{theorem}
\begin{proof}
Since $\dot{E}$ is fixed
it is sufficient to prove the first
claim for each operator separately. 
Note that since $\dot{E}$ is fixed,
it follows that $n$ is also fixed. 
\begin{enumerate}
\item
If $\dot{E}(\dot{\mathbf{C}}) = \dot{\mathsf{C}}_p(\dot{\mathbf{C}})$, 
then 
$|\dot{E}(\dot{\mathbf{C}})|
\,=\,
|\dot{C}_p|
\,\leq\,
|\dot{\mathbf{C}}|
\,=\, 
\mathcal{O}(\pi(|\dot{\mathbf{C}}|))$.
	
\item
If $\dot{E}(\dot{\mathbf{C}}) 
= 
\dot{\mathsf{d}}_{ij}(\dot{\mathbf{C}})$,
then 
$|\dot{E}(\dot{\mathbf{C}})|
\,=\,
\mathcal{O}(|\dot{\mathbf{C}}|) 
\times
\mathcal{O}(1) 
\,=\, 
\mathcal{O}(\pi(|\dot{\mathbf{C}}|))$.

\item
If 
$
\dot{E}(\dot{\mathbf{C}})  
= 
\dot{\mathsf{C}}_p(\dot{\mathbf{C}})\cup
\dot{\mathsf{C}}_q(\dot{\mathbf{C}}),
$ 
then
$|\dot{E}(\dot{\mathbf{C}})|
\,\leq\, 
|\dot{\mathbf{C}}| 
\,=\, 
\mathcal{O}(\pi(|\dot{\mathbf{C}}|))$.
	
\item
If $\dot{E}(\dot{\mathbf{C}})  
= 
\dot{\mathsf{C}}_p(\dot{\mathbf{C}})\capdot
\dot{\mathsf{C}}_q(\dot{\mathbf{C}})
$,
then the number of tuples 
in $\dot{E}(\dot{\mathbf{C}})$   
is at most
$|\dot{\mathbf{C}}|^2$,
and each tuple in the output can have a condition of 
length at most 
$2\cdot|\dot{\mathbf{C}}|$.
As a result,
$|\dot{E}(\dot{\mathbf{C}})|
\,\leq\, 
2\cdot|\dot{\mathbf{C}}|^3  
\,=\,
\mathcal{O}(\pi(|\dot{\mathbf{C}}|))$.

\item
If 
$
\dot{E}(\dot{\mathbf{C}})  
= 
\dot{\mathsf{c}}_i(\mathsf{C}_p(\dot{\mathbf{C}}))$, then 
$|\dot{E}(\dot{\mathbf{C}})|
\,\leq\, 
|\dot{C}_p|
\,\leq\, 
|\dot{\mathsf{\mathbf{C}}}|
\,=\, 
\mathcal{O}(\pi(|\dot{\mathbf{C}}|))$.
		
\item
For the case
$\dot{E}(\dot{\mathbf{C}})  
= 
\dot{\cdown}_i(\mathsf{C}_p(\dot{\mathbf{C}}))$
we note that
$\dot{\cdown}_i(\mathsf{C}_p(\dot{\mathbf{C}}))
\subseteq 
(\mathsf{C}_p(\dot{\mathbf{C}}) \capdot\dot{A})
\subseteq 
\dot{A}$. 
We can construct the star-tuples in $\dot{A}$
by iterating over the star-tuples in $\dot{C}_p$
and using the constants in $A$. This means that
$|\dot{A}|=
(n \times (|A|))+ (2^n + |A|) \leq
\mathcal{O}(1) \times \mathcal{O}(|\dot{\mathbf{C}}|)
+ \mathcal{O}(1) \times \mathcal{O}(|\dot{\mathbf{C}}|)
=
\mathcal{O}(\pi(|\dot{\mathbf{C}}|))$.
Note that $n$ is the dimensionality of
$\dot{\mathbf{C}}$ and is a constant.
		
\item 
If 
$\dot{E}(\dot{\mathbf{C}})  
= 
\boldsymbol{\dot{\neg}}(\mathsf{C}_p(\dot{\mathbf{C}}))$,
then similar to the inner cylindrification
we have
$\boldmath{\dot{\neg}}\,\dot{C}_p \subseteq \dot{A}$
which implies 
$|\dot{E}(\dot{\mathbf{C}})| =  
\mathcal{O}(\pi(|\dot{\mathbf{C}}|))$.
\end{enumerate}
	
For the time complexity we note 
that we need to
keep all conditions
in $\dot{E}(\mathbf{C})$ as a logically closed
set of literals.
To do this,
we start with $\dot{C}$ in normal form.
With each tuple we associate a graph with
vertices $\{1,\ldots,n\}$ and a 
blue edge $\{i,j\}$
if $(i=j)$ is in the condition of the tuple,,
and a red edge $\{i,j\}$
if $(i\neq j)$ is in the condition.
Next, we compute the transitive closure
of the graph wrt the blue edges.
Then each connected component is an equivalence class,
unless there is some pair $\{i,j\}$ that  has
both a blue and a red edge, in which case
the tuple is inconsistent.
We don't need to consider conditions of the forms
$(i\neq a)$ in star-tuples $\dot{t}$.
This is because then $\dot{t}(i)=*$,
and $\dot{t}(i)$ can still contain a 
an infinite (more precisely cofinite)
set of values
Finally, when we compute the dot-intersection,
for each pair of tuples, we take
the union of their condition graphs, and
recompute the blue transitive closure.
Since $\dot{E}$ is fixed,
this needs to be done only a bounded
number of times.
\end{proof}

\bigskip
\noindent
\textbf{Membership.}
In the membership problems,
we ask if an ordinary tuple $t$ 
belongs to the set specified by
a (naive) star-cylinder,
of by a fixed expression $\dot{E}$
and a (naive) star-cylinder.
In other words, all results
refer to data complexity.

\begin{theorem}\label{membership}
Let $t\in\mathbb{D}^n$
and $\ddot{\mathbf{C}}$ 
a sequence of 
$n$-dimensional naive star-cylinders and diagonals.
The membership problems and their 
respective data complexities are as follows.
\begin{enumerate}
	
\item
$t \stackrel{?}{\in} \bigcap E(\rep(\lsem\ddot{\mathbf{C}}\rsem))$ 
is in polytime for positive $E$.

\item
$t \stackrel{?}{\in}  \bigcap E(\rep(\lsem\ddot{\mathbf{C}}\rsem))$ 
is coNP-complete for $E$ where
negation is allowed in equality atoms only.
\end{enumerate}
\end{theorem}
\begin{proof}

$\;$

\begin{enumerate}		
\item
By Theorem~\ref{main2}, we have
$\bigcap E(\rep(\lsem\ddot{\mathbf{C}}\rsem))$
$= 
\lsem\dot{E}(\ddot{\mathbf{C}})_{\downarrow}\rsem$,
so to test if
$t \in \bigcap E(\rep(\lsem\ddot{\mathbf{C}}\rsem))$,
we compute 		
$\dot{E}(\ddot{\mathbf{C}})_{\downarrow}$,
and see if there is a star-tuple
$\dot{t} \in \dot{E}(\ddot{\mathbf{C}})_{\downarrow}$,
such that $t\preceq\dot{t}$.
By Theorem \ref{CA-complexity},
$\dot{E}(\ddot{\mathbf{C}})_{\downarrow}$
can be computed in polytime. 		

\item
To check if
$t \not \in \bigcap E(\rep(\lsem\ddot{\mathbf{C}}\rsem))$,
it is sufficient to find a 
homomorphism $h$ such that
$t \not \in h(\lsem\ddot{\mathbf{C}}\rsem)$.
We guess the homomorphism $h$,
and check in polytime if
$t \not \in h(\lsem\ddot{\mathbf{C}}\rsem)$.
Thus $t \not \in \bigcap E(\rep(\lsem\ddot{\mathbf{C}}\rsem))$ 
is in NP, and 
$t \in \bigcap E(\rep(\lsem\ddot{\mathbf{C}}\rsem))$ is in coNP.
The lower bound follows from Theorem 5.2.2 
in \cite{DBLP:journals/tcs/AbiteboulKG91}.
\end{enumerate}
\end{proof}

\bigskip
\noindent
{\bf Containment.}
The containment problems ask 
for containment
of star-cylinders
(naive star-cylinders),
or views over star-cylinders
(naive star-cylinders).
We have the following.
\begin{theorem}\label{containment}
Let
$\dot{\mathbf{C}}$ and $\dot{\mathbf{D}}$
(resp.\ $\ddot{\mathbf{C}}$ and $\ddot{\mathbf{D}}$)
be sequences of $n$-dimensional 
(naive) star-cylinders and diagonals. 
Then
\begin{enumerate}
\item
$
E_1(\lsem\dot{\mathbf{C}}\rsem) 
\;\stackrel{?}{\subseteq}\; 
E_2(\lsem\dot{\mathbf{D}}\rsem)
$ 
is in polytime for $\mbox{CA}_n$ expression $E_1$ and $E_2$.

\item
$
\rep(\lsem\ddot{\mathbf{C}}\rsem)
\;\stackrel{?}{\subseteq}\; 
\rep(\lsem\ddot{\mathbf{D}}\rsem)$
is NP-complete.

\item
$
E_1(\rep(\lsem\ddot{\mathbf{C}}\rsem)) 
\;\stackrel{?}{\subseteq}\; 
E_2(\rep(\lsem\ddot{\mathbf{D}}\rsem)) 
$
is $\Pi^p_2$-complete
for positive $E_1$ and $E_2$.
\end{enumerate}
\end{theorem}
\begin{proof}
	$\;$
	
\begin{enumerate}
\item 	
By Lemma \ref{star-cyl-dom}, we have
$
\lsem\dot{E}_1(\dot{\mathbf{C}})\rsem 
\subseteq 
\lsem\dot{E}_2(\dot{D})\rsem
$ 
if and only if
$\dot{E}_1(\dot{\mathbf{C}})\capdot\dot{A} 
\,\preceq\, 
\dot{E}_2(\dot{D})\capdot\dot{A}$. 
The latter dominance is true
if and only if
for each star-tuple 
$\dot{t}\in\dot{E}_1(\dot{\mathbf{C}})\capdot\dot{A}$
there is a star-tuple
$\dot{u}\in\dot{E}_2(\dot{\mathbf{D}})\capdot\dot{A}$,
such that $\dot{t}\preceq\dot{u}$.
From Theorem~\ref{CA-complexity}
we know that $\dot{E}_1(\dot{\mathbf{C}})\capdot\dot{A}$
and $\dot{E}_2(\dot{\mathbf{D}})\capdot\dot{A}$
can be computed in polytime.

\item
We first 
extend the domain of possible world
homomorphisms by stipulating that
they are the identity on $\ast$.
Then it is easy to see that
$\rep(\lsem\ddot{\mathbf{C}}\rsem)
\subseteq \rep(\lsem\ddot{\mathbf{D}}\rsem)$ 
if and only if 
there exists 
a possible world homomorphism 
$h$ such that
$\ddot{\mathbf{D}} \rightarrow_h \ddot{\mathbf{C}}$.
This makes the problem NP-complete.
		
\item
The lower bound follows from
Theorem 4.2.2 in
\cite{DBLP:journals/tcs/AbiteboulKG91},
For the upper bound we observe that
$E_1(\rep(\lsem\ddot{\mathbf{C}}\rsem))$  
$\subseteq 
E_2(\rep(\lsem\ddot{\mathbf{D}}\rsem)$ 
iff for every
$\mathbf{C} \in \rep(\lsem\ddot{\mathbf{C}}\rsem)$ 
there exists a 
$\mathbf{D} \in \rep(\lsem\ddot{\mathbf{D}}\rsem)$ such that 
$E_1(\mathbf{C}) = E_2(\mathbf{D})$
iff for every possible world homomorphism $h$ on
$\ddot{\mathbf{C}}$ there exists
a possible world homomorphism $g$ on
$\ddot{\mathbf{D}}$ 
such that 
$E_1(h(\lsem\ddot{\mathbf{C}}\rsem)) 
= 
E_2(g(\lsem\ddot{\mathbf{D}}\rsem))$. 
By Corollary \ref{corocacastar},
this equality holds iff
$\lsem\dot{E}_1(h(\ddot{\mathbf{C}}))\rsem
= 
\lsem\dot{E}_1(g(\ddot{\mathbf{D}}))\rsem$.
By Lemma \ref{star-cyl-dom}, 
the last equality holds iff		
$E_1(h(\lsem\ddot{\mathbf{C}}\rsem))\capdot\dot{A}
\,\preceq\, 
E_2(g(\lsem\ddot{\mathbf{D}}\rsem))\capdot\dot{A}$,
and vice-versa.
By 
Theorem \ref{CA-complexity},
the star-cylinders in the two
dominances $\preceq$
can be computed in polynomial time.
\end{enumerate}
\end{proof}

%
\section{Related and future work}\label{related}

Cylindric Set Algebra gave rise to
a whole subfield of Algebra, called
Cylindric Algebra. For a fairly recent overview,
the reader is referred to \cite{tarskipic}.
Within database theory, a simplified version
of the star-cylinders and a corresponding
Codd-style positive relational algebra with
evaluation rules ``$\ast=\ast$'' and ``$\ast~=~a$''
was proposed by Imielinski and Lipski
in \cite{ilcyl}. Such cylinders correspond to
the structures in {\em diagonal-free} 
Cylindric Set Algebras
\cite{hmt1,hmt2}.
The exact FO-expressive power of
these diagonal-free star-cylinders
is an open question. Nevertheless,
using the techniques of this paper,
it can be shown that naive existential
nulls can be seamlessly incorporated 
in diagonal-free star-cylinders.

In addition  to the above and the work
described in Section \ref{intro},
Imielinski and Lipski 
also showed in 
\cite{ilcyl} 
that the fact that Codd's Relational Algebra
does not have a finite axiomatization,
and the fact that equivalence of expressions in it
is undecidable, 
follow from known results
in Cylindric Algebra. This is of course
true for a host of general results
in Mathematical Logic. 

Yannakakis and Papadimitriou \cite{DBLP:journals/jcss/YannakakisP82}
formulated an algebraic version of
dependency theory
using Codd's Relational Algebra.
Around the same time Cosmadakis~\cite{cos87}
proposed an interpretation of 
dependency theory in terms of
equations over certain types
of expressions in Cylindric Set Algebra,
and described a complete finite 
axiomatization of his system.
It was however later shown by
D\"{u}ntsch, Hodges, and Mikulas 
\cite{hm,dm},
again using known results from Cylindric Algebra,
that Cosmadakis's axiomatization
was incomplete, and that no
finite complete axiomatization exists.  

Interestingly, it turns out that
one of the models for constraint databases
in \cite{kkr} by Kanellakis, Kuper, and Revesz
--- the one where the constraints are equalities
over an infinite domain ---
is equivalent with our star-tables.
Even though \cite{kkr} develops a
bottom-up (recursive) evaluation mechanism
for FO-queries,
the mechanism is goal-oriented and 
contrary to our star-cylinders,
there is no algebra operating
on the constraint databases.
We note however that the construction
of the sieve $\dot{A}$ in Section \ref{caandcastar} 
is inspired by the constraint solving techniques
of \cite{kkr}. It therefore seems
that our star-cylinders and algebra can be made to
handle inequality constraints on 
dense linear orders as well as 
polynomial constraints over real-numbers,
as is done in \cite{kkr}.
We also note that our work is related to the orbit
finite sets, treated in a general computational
framework in~\cite{orbit}.

As noted in Section \ref{intro},
the existential nulls have long been
well understood.
According to \cite{DBLP:conf/pods/DeutschNR08}
the fact that positive queries
(no negation, but allowing universal quantification)
are preserved under onto-homomorphisms
are folklore in the database community.
Using this monotonicity property,
Libkin \cite{libkin-naive} has recently
shown that positive queries can be
evaluated naively on finite existential
databases $I$ under a so called
{\em weak closed world assumption},
where $\rep(I)$ consists of all
complete instances $J$,
such that $h(I)\subseteq J$
and $J$ only involves constants that occur in $I$,
and furthermore $h$ is onto from the finite universe of $I$ to
the finite universe of $J$.
Our Theorem \ref{naivemain} generalizes Libkin's
result to infinite databases.
In this context it is worth noting that
Lyndon's Positivity Theorem
\cite{lyndon}
tells us that a
first order formula is preserved
under onto-homomorphisms on all structures
if and only if it is equivalent to a positive
formula.
It has subsequently been shown
that the only-if direction fails
for finite structures \cite{gur,stolb}.  
Since our star-cylinders represent
neither finite nor unrestricted infinite
structures, it would be interesting to
know whether the only-if direction holds
for infinite structures represented by star-cylinders.
If it does, it would mean that our Theorem \ref{naivemain}
would be optimal, meaning that if $\varphi$
is not equivalent to a positive formula,
then naive evaluation does not work for
$\varphi$ on databases represented by
naive star-cylinders.

Finally we note that Sundarmurthy et al.\ 
\cite{DBLP:conf/icdt/SundarmurthyKLN17}
have generalized the conditional tables
of \cite{DBLP:journals/jacm/ImielinskiL84,grahne91} by replacing the 
labelled nulls with a single null $\mathbf{m}$
that initially represents all possible domain values.
They then add constraints on the occurrences of these
$\mathbf{m}$-values, allowing them to represent
a finite or infinite subset of the domain,
and to equate distinct occurrences of $\mathbf{m}$.
Sundarmurthy et al.\ 
then show that their $\mathbf{m}$-tables are
closed under positive (but not allowing universal
quantification) queries by developing a difference-free
Codd-style relational algebra that
$\mathbf{m}$-tables are closed under.
Merging our approach with theirs could
open up interesting possibilities.

\bibliography{Bibliography}
\end{document}